\documentclass{jfp1-tweaked}
\usepackage{etex}

\usepackage{tabvar}

\usepackage[utf8]{inputenc} 

\usepackage{amsmath}
\usepackage{amssymb}
\usepackage{stmaryrd}
\usepackage{graphicx}
\usepackage{nameref,hyperref,cleveref}
\usepackage{color}
\usepackage{tikz}
\usetikzlibrary{positioning}
\usetikzlibrary{decorations.markings}

\definecolor{darkgreen}{RGB}{20,200,10}

\usepackage[normalem]{ulem}

\usepackage{amsthm}
\usepackage{proof}

\usepackage{thmtools}

\usepackage[all,2cell]{xy}
\UseAllTwocells

\usepackage{multicol}

\newtheorem{theorem}{Theorem}[section]

\newtheorem{corollary}[theorem]{Corollary}

\newtheorem{proposition}[theorem]{Proposition}
\newtheorem{claim}[theorem]{Claim}

\newtheorem{definition}[theorem]{Definition}

\newtheorem{example}{Example}

\crefname{theorem}{thm.}{theorems}
\crefname{corollary}{cor.}{corollaries}
\crefname{proposition}{prop.}{propositions}
\crefname{claim}{claim}{claims}
\crefname{observation}{obs.}{observations}

\theoremstyle{definition}
\crefname{definition}{defn.}{definitions}

\newcommand\etal{\emph{et al.}}

\newcommand\definand[1]{{\bf #1}}

\newcommand\set[1]{\{\,#1\,\}}

\newcommand\GFlin{L}
\newcommand\GFpla{P}

\newcommand\GFlinind{\GFlin_{\mathrm{ind}}}
\newcommand\GFplaind{\GFpla_{\mathrm{ind}}}

\newcommand\trip[3]{(#1~#2~#3)}
\newcommand\tup[2]{(#1~#2)}

\newcommand\cartgp{\mathcal{C}}
\newcommand\trigp{\mathcal{T}}
\newcommand\freegp[1]{\left\langle #1\right\rangle}

\newcommand\defeq{\mathbin{\overset{\mathrm{def}}{=}}}

\newcommand\ImpL{\multimap}
\newcommand\ImpR{\mathbin{\text{\reflectbox{$\multimap$}}}}

\newcommand\impL[2][\bot]{{#2} \mathbin{\ImpL} {#1}}
\newcommand\impR[2][\bot]{{#1} \mathbin{\ImpR} {#2}}

\newcommand\sem[1]{\left\llbracket #1\right\rrbracket}

\newcommand\op{\mathrm{op}}
\newcommand\Cat{\mathbf{Cat}}

\newcommand\gset[1]{#1\text{-set}}
\newcommand\fixpt[1]{\mathrm{fix}_{#1}}

\newcommand\llift[2]{\lambda #1[#2]}

\def\Icomb{\mathrm{I}}
\def\Bcomb{\mathrm{B}}
\def\Ccomb{\mathrm{C}}

\title[Journal of Functional Programming]{Linear lambda terms as invariants \\ of rooted trivalent maps}
\author[N. Zeilberger]{Noam Zeilberger \\ Inria, Campus de l'École Polytechnique, 91120 Palaiseau, France \\ \email{noam.zeilberger@gmail.com}}

\begin{document}
\maketitle[t]

\begin{abstract}
The main aim of the article is to give a simple and conceptual account for the correspondence (originally described by Bodini, Gardy, and Jacquot) between $\alpha$-equivalence classes of closed linear lambda terms and isomorphism classes of rooted trivalent maps on compact oriented surfaces without boundary, as an instance of a more general correspondence between linear lambda terms with a context of free variables and rooted trivalent maps with a boundary of free edges.
We begin by recalling a familiar diagrammatic representation for linear lambda terms, while at the same time explaining how such diagrams may be read formally as a notation for endomorphisms of a reflexive object in a symmetric monoidal closed (bi)category.
From there, the ``easy'' direction of the correspondence is a simple forgetful operation which erases annotations on the diagram of a linear lambda term to produce a rooted trivalent map.
The other direction views linear lambda terms as \emph{complete invariants} of their underlying rooted trivalent maps, reconstructing the missing information through a Tutte-style topological recurrence on maps with free edges.
As an application in combinatorics, we use this analysis to enumerate bridgeless rooted trivalent maps as linear lambda terms containing no closed proper subterms, and conclude by giving a natural reformulation of the Four Color Theorem as a statement about typing in lambda calculus.
\end{abstract}
\section{Introduction}
\label{sec:intro}

This paper follows recent work on the combinatorics of linear lambda terms, which has uncovered various connections to the theory of graphs on surfaces (or ``maps'').
It is currently known that there exist size-preserving correspondences between all of the following pairs of families of objects, some with explicit bijections but all at least at the level of generating functions \cite{bodini-et-al,zg2015rpmnpt,z2015counting}:
\begin{center}
\begin{oldtabular}{|@{\,}c@{\,}|@{\,}c@{\,}|@{\,}c@{\,}|}
\hline 
Family of rooted maps & Family of lambda terms  & OEIS entry
\\
\hline\hline
rooted trivalent maps & linear lambda terms & \href{https://oeis.org/A062980}{A062980}\\
\hline
rooted planar maps & normal planar lambda terms & \href{https://oeis.org/A000168}{A000168} \\
\hline
rooted maps & normal linear lambda terms / $\sim$ & \href{https://oeis.org/A000698}{A000698} \\
\hline
\end{oldtabular}
\end{center}
(Here ``OEIS'' is short for \emph{The On-Line Encyclopedia of Integer Sequences} \cite{oeis}.)
Although the existence of such connections is intriguing, it is not yet obvious to what extent they have a deeper ``meaning''.
My aim in this article, therefore, is to revisit the basic situation of rooted trivalent maps and propose a slightly more general and conceptual account of the bijection originally given by Bodini, Gardy, and Jacquot -- an account which hopefully suggests some clear directions for further exploration.
The main insights I hope to convey are that:
\begin{enumerate}
\item Bodini~\etal's bijection between \emph{closed} linear lambda terms and rooted trivalent maps \emph{on compact oriented surfaces without boundary} is really an instance of a more general bijection that relates linear lambda terms with free variables to rooted trivalent maps with a marked boundary of free edges.
\item If we represent linear lambda terms using a natural diagrammatic syntax, then the corresponding rooted trivalent maps are obtained simply by erasing some information stored locally at the nodes of the diagram.
  Moreover, through a little bit of category theory, this way of representing linear lambda terms (which is folklore) can be understood within the wider context of \emph{string diagrams}, as a notation for endomorphisms of a reflexive object in a symmetric monoidal closed (bi)category.
\item Conversely, by considering connectivity properties of the underlying graph, it is possible to invert this forgetful transformation through a recursive decomposition of rooted trivalent maps with free edges (similar in spirit to Tutte's seminal analysis of rooted planar maps \cite{tutte1968}).  In effect, a linear lambda term can be seen as a \emph{topological invariant} of a rooted trivalent map (analogous to, say, the chromatic polynomial of a graph), which is moreover a complete invariant in the sense that it characterizes the rooted trivalent map up to isomorphism.
\end{enumerate}
One immediate application of this analysis will be a simple characterization of rooted trivalent maps without bridges, as linear lambda terms with no closed proper subterms.
I will then show how to combine this characterization with a link between typing and graph-coloring, to yield a surprising yet natural lambda calculus reformulation of the Four Color Theorem (in its equivalent form as the statement that every bridgeless trivalent planar map is edge 3-colorable).

The rest of the paper is structured as follows.
\Cref{sec:review-maps,sec:review-lambda} provide some elementary background on maps and lambda calculus, while \Cref{sec:stringdiagrams} explains how to represent linear lambda terms graphically and how to interpret these diagrams categorically.
The bijection between linear lambda terms with free variables and rooted trivalent maps with free edges (which extends the bijection of \cite{bodini-et-al} on closed terms) is presented in \Cref{sec:linear-trivalent,sec:trivalent-linear}.
Finally, \Cref{sec:4ct} discusses the characterization of bridgeless rooted trivalent maps, and the reformulation of the Four Color Theorem (4CT).

\section{Classical definitions for rooted trivalent maps}
\label{sec:review-maps}

This section recalls some standard definitions from the theory of maps (for further background on the subject, see Lando and Zvonkin \shortcite{landozvonkin}).
In topological terms, a map can be defined as a 2-cell embedding $i : G \hookrightarrow X$ of an undirected graph $G$ (loops and multiple edges allowed) into a surface $X$: that is, a representation of the vertices $v \in G$ by points $i(v) \in X$ and the edges $v_1 \overset{e}{\leftrightarrow} v_2 \in G$ by arcs $i(v_1) \overset{i(e)}\frown i(v_2) \in X$, such that no two arcs cross, and such that the complement of the graph inside the surface $X \setminus i(G)$ is a disjoint union of simply-connected regions (called \emph{faces}; note that this last condition implies that if the underlying surface $X$ is connected, then $G$ must be a connected graph).
Two maps $G \overset{i}\hookrightarrow X$ and $G' \overset{i'}\hookrightarrow X'$ are said to be \emph{isomorphic} if there is a homeomorphism between the underlying surfaces $h : X \to X'$ whose restriction $h|_{i(G)}$ witnesses an isomorphism of graphs $G \to G'$.

One of the beautiful aspects of the theory is that in many situations, maps also admit a purely algebraic description as a collection of permutations satisfying a few properties.
For example, a 2-cell embedding of a graph into any \emph{compact oriented surface without boundary} \cite{jones-singerman} may be represented as a pair of permutations $v$ and $e$ on a set $M$ such that
\begin{enumerate}
\item $e$ is a fixed point-free involution, and
\item the group $\freegp{v,e}$ generated by the permutations acts transitively on $M$ (i.e., for any pair of elements $x,y\in M$ it is possible to go from $x$ to $y$ by some sequence of applications of $v$ and $e$).
\end{enumerate}
This kind of representation is sometimes called a \emph{combinatorial map}.
The idea is that the elements of the set $M$ stand for ``darts'' or ``half-edges'', so that the involution $e$ describes the gluing of pairs of darts to form an edge, while the permutation $v$ encodes the cyclic (say, counterclockwise) ordering of darts around each vertex.
\Cref{fig:example-map} gives an example of a planar map (i.e., an embedding of a graph into the sphere $X = S^2$) represented by such permutations.
In addition to the \emph{vertex permutation} $v$ and the \emph{edge permutation} $e$, to any combinatorial map one may associate a \emph{face permutation} $f$ by the equation $f = (ev)^{-1} = v^{-1}e$, representing the cyclic ordering of darts around each face of the corresponding embedded graph.
\begin{figure}
  \begin{center}
  \begin{minipage}{0.2\textwidth}
  \begin{tikzpicture}
    \node (a) at (0,0) {$\bullet$};
    \node (b) at (-1,1) {$\bullet$};
    \node (c) at (0.1,1.5) {$\bullet$};
    \node (d) at (2,1.6) {$\bullet$};
    \node (e) at (1.8,0.2) {$\bullet$};
    \node (f) at (0.2,2.5) {$\bullet$};
    \node (g) at (2.4,-0.2) {};
    \draw (a.center) to node [pos=0.2] {\tiny$3$} (b.center);
    \draw (b.center) to node [pos=0.2] {\tiny$4$}  (a.center);
    \draw (a.center) to node [pos=0.3] {\tiny$2$} (c.center);
    \draw (c.center) to node [pos=0.2] {\tiny$14$} (a.center);
    \draw (b.center) to node [pos=0.2] {\tiny$5$} (c.center);
    \draw (c.center) to node [pos=0.2] {\tiny$13$} (b.center);
    \draw (c.center) to node [pos=0.2] {\tiny$12$} (d.center);
    \draw (d.center) to node [pos=0.2] {\tiny$10$} (c.center);
    \draw (d.center) to node [pos=0.2] {\tiny$11$} (e.center);
    \draw (e.center) to node [pos=0.2] {\tiny$15$} (d.center);
    \draw (e.center) to node [pos=0.15] {\tiny$16$} (a.center);
    \draw (a.center) to node [pos=0.15] {\tiny$1$} (e.center);
    \draw (b.center) to node [pos=0.2] {\tiny$6$} (f.center);
    \draw (f.center) to node [pos=0.2] {\tiny$7$} (b.center);
    \draw (f.center) to node [pos=0.2] {\tiny$8$} (d.center);
    \draw (d.center) to node [pos=0.2] {\tiny$9$} (f.center);
    \draw [bend right=45,rounded corners=5pt] (e.center) to node [pos=0.3] {\tiny$17$} (g.center) to node [pos=0.7] {\tiny$18$} (e.center);
  \end{tikzpicture}
  \end{minipage}
  \qquad\qquad
  \begin{minipage}{0.6\textwidth}
  \begin{align*}
    M &= \set{1,\dots,18} \\
    v &= \trip 123\trip 456\tup 78\trip{9}{10}{11}\trip{12}{13}{14}(15~16~17~18) \\
    e &= \tup 1{16}\tup 2{14}\tup 34\tup 5{13}\tup 67\tup 89\tup{10}{12}\tup {11}{15}\tup{17}{18} \\
    f &= (1~15~10~14)(2~13~4)(3~6~8~11~18~16)(5~12~9~7)(17)
  \end{align*}
  \end{minipage}
  \end{center}
  \caption{A planar map represented by permutations}
  \label{fig:example-map}
\end{figure}
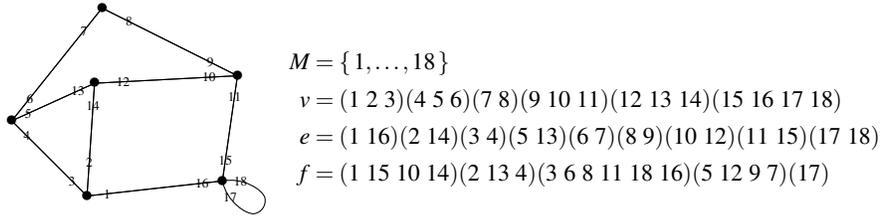
Two combinatorial maps $(M,v,e)$ and $(M',v',e')$ are considered as isomorphic if there is a bijection between the underlying sets $h : M \to M'$ which is compatible with the action of the permutations, $hv = v'h$, $he = e'h$.
The following definition rephrases all of this a bit more efficiently:
\begin{definition}
  Let $\cartgp$ be the group $\cartgp \defeq \freegp{v,e \mid e^2 = 1}$.  A (combinatorial) \definand{map} (on a compact oriented surface without boundary) is a transitive $\gset\cartgp$ on which the generator $e$ acts without fixed points.
\end{definition}
\noindent
One nice feature of combinatorial maps is that it is easy to compute their \emph{genus}.
\begin{definition}
Let $M$ be a combinatorial map.
The \definand{genus} $g$ of $M$ is defined by the Euler-Poincar\'e formula $c(v) - c(e) + c(f) = 2 - 2g$, where $v$, $e$, and $f$ are respectively the vertex, edge, and face permutations associated to $M$, and $c(\pi)$ counts the number of cycles in the cycle decomposition of $\pi$.
(For example, for the genus $g=0$ map of \Cref{fig:example-map} we have $c(v) - c(e) + c(f) = 6 - 9 + 5 = 2$.)
\end{definition}
\noindent
In the rest of the paper we will be focused on \emph{trivalent} maps (also called cubic maps), which can be defined by the algebraic condition that the vertex permutation is fixed point-free and of order three.
\begin{definition}
\label{def:trivalent}
  Let $\trigp$ be the group $\trigp \defeq \freegp{v,e \mid v^3 = e^2 = 1}$.  A \definand{trivalent map} is a transitive $\gset\trigp$ on which the generators $v$ and $e$ act without fixed points.
\end{definition}
\noindent
Moreover, we will always be speaking about so-called ``rooted'' maps.
\begin{definition}
\label{def:rooted}
A \definand{rooted} (trivalent) map is a (trivalent) map $M$ with a distinguished element $r \in M$ called the \uline{root}.
An isomorphism of rooted maps $f : (M,r) \to (M',r')$ is an isomorphism of maps $f : M \to M'$ which preserves the root $f(r) = r'$.
\end{definition}
\noindent
Topologically, a rooting of a map can be described as the choice of an edge, vertex, and face all mutually incident, or equivalently as the choice of an edge together with an orientation of that edge.
The original motivation for the study of rooted maps in combinatorics was that they are \emph{rigid} objects (meaning that they have no automorphisms other than the identity) and hence are easier to count than unrooted maps,\footnote{For a historical account, see Ch.~10 of Tutte's \emph{Graph Theory as I Have Known it} (1998, Oxford).} but it is also worth remarking that there is a general correspondence
\begin{center}
pointed transitive $G$-sets $\leftrightarrow$ subgroups of $G$
\end{center}
which sends any transitive $\gset{G}$ $M$ equipped with a distinguished point $r \in M$ to the stabilizer subgroup $G_r = \set{g \in G \mid g \ast r = r}$, and any subgroup $H \subseteq G$ to the action of $G$ on cosets $G/H$ together with the distinguished coset $H$.
In particular, every rooted trivalent map uniquely determines a subgroup of the modular group $\trigp \cong \mathrm{PSL}(2,\mathbb{Z})$, while an unrooted map only determines one up to conjugacy \cite{jones-singerman94schneps,vidal2010}.

Finally, it is fairly common (cf.~\cite{jones-singerman94schneps,vidal2010}) to relax the conditions of fixed point-freeness in the definition of a general map and/or of a trivalent map.
Intuitively, fixed points of $e$ represent ``dangling'' edges, while fixed points of $v$ represent univalent vertices.
Precise formulations differ, however, and in \Cref{sec:linear-trivalent} we will introduce a generalization of the classical definition of rooted trivalent maps which is motivated by the correspondence with linear lambda terms.

\section{Basic definitions for linear lambda terms}
\label{sec:review-lambda}

Here we cover the small amount of background on lambda calculus that we will need in order to talk about linear lambda terms (for a more general introduction, see Barendregt \shortcite{barendregt1984}).
The terms of pure lambda calculus are constructed from variables ($x,y,\dots$) using only the two basic operations of \emph{application} $t(u)$ and \emph{abstraction} $\lambda x[t]$.
Within a given term, one distinguishes \emph{free} variables from \emph{bound} variables.
An abstraction $\lambda x[t]$ is said to bind the occurrences of $x$ within the subterm $t$, and any variable which is not bound by an abstraction is said to be free.
Two lambda terms are considered equivalent (``$\alpha$-equivalent'') if, roughly speaking, they differ only by renaming of bound variables.

A term is said to be \emph{linear} if every variable (free or bound) has exactly one occurrence: for example, the terms $\lambda x[x(\lambda y[y])]$, $\lambda x[\lambda y[x(y)]]$, and $\lambda x[\lambda y[y(x)]]$ are linear, but the terms $\lambda x[x(x)]$, $\lambda x[\lambda y[x]]$, and $\lambda x[\lambda y[y]]$ are non-linear.
To make this definition more precise, it is natural to consider linear lambda terms as indexed explicitly by lists of free variables, called \emph{contexts}, which also affects the definition of $\alpha$-equivalence.
\begin{definition}\label{def:linear}
A \definand{context} is an ordered list of distinct variables $\Gamma = (x_1,\dots,x_k)$.
We write $(\Gamma,\Delta)$ for the concatenation of two contexts $\Gamma$ and $\Delta$, with the implicit condition that they contain disjoint sets of variables.
Let $\Gamma\vdash t$ be the relation between contexts and lambda terms defined inductively by the following rules:
\begin{equation}
\infer{x \vdash x}{}
\qquad
\infer{\Gamma,\Delta \vdash t(u)}{\Gamma \vdash t & \Delta \vdash u}
\qquad
\infer{\Gamma \vdash \lambda x[t]}{\Gamma, x \vdash t}
\qquad\qquad
\infer{\Gamma, x,y,\Delta \vdash t}{\Gamma, y,x,\Delta \vdash t}
\label{eqn:linrules}
\end{equation}
Then a \definand{linear lambda term} is a pair $(\Gamma,t)$ of a context and a term such that $\Gamma \vdash t$.
Two linear lambda terms $(\Gamma,t)$ and $(\Gamma',t')$ are said to be \definand{$\alpha$-equivalent} if they only differ by a series of changes of free or bound variables, defined as follows.
Supposing that $\Gamma = (\Gamma_1,x,\Gamma_2)$ and $\Gamma' = (\Gamma_1,y,\Gamma_2)$, then $(\Gamma,t)$ and $(\Gamma',t')$ differ by a single \definand{change of free variable} if $t' = t\{y/x\}$, where $t\{y/x\}$ denotes the substitution of $y$ for $x$ in $t$.
Similarly, $(\Gamma,t)$ and $(\Gamma',t')$ differ by a single \definand{change of bound variable} if $t'$ arises from $t$ by replacing some subterm $\lambda x[u]$ by $\lambda y[u\{y/x\}]$.
\end{definition}
\noindent
For combinatorists, it might be helpful to see the two-variable generating function counting $\alpha$-equivalence classes of linear lambda terms of a given size in a given context.
Let $t_{n,k}$ stand for the number of $\alpha$-equivalence classes of linear lambda terms with $n$ total applications and abstractions and with $k$ free variables.
Then the generating function
$\GFlin(z,x) = \sum_{n,k} t_{n,k} \frac{x^k z^n}{k!}$
satisfies the following functional-differential equation:
\begin{equation}
\GFlin(z,x) = x + z\GFlin(z,x)^2 + z\frac{\partial}{\partial x}\GFlin(z,x)
\label{eqn:lingf}
\end{equation}
The three summands in equation (\ref{eqn:lingf}) correspond to the three rules on the left side of (\ref{eqn:linrules}), 
$$
\infer{x \vdash x}{}
\qquad
\infer{\Gamma,\Delta \vdash t(u)}{\Gamma \vdash t & \Delta \vdash u}
\qquad
\infer{\Gamma \vdash \lambda x[t]}{\Gamma, x \vdash t}
$$
while the rule on the right side of (\ref{eqn:linrules})
$$
\infer{\Gamma, x,y,\Delta \vdash t}{\Gamma, y,x,\Delta \vdash t}
$$
explains why the generating function $\GFlin(z,x)$ is of \emph{exponential type} in the parameter $x$: if $(\Gamma,t)$ is a linear lambda term, then so is $(\Gamma',t)$ for any permutation $\Gamma'$ of $\Gamma$.
Finally, note that instantiating $\GFlin(z,0)$ gives the ordinary generating function (OGF) counting \emph{closed} linear lambda terms by total number of applications and abstractions,
$$
L(z,0) = z + 5z^3 + 60z^5 + 1105z^7 + 27120z^9 + \dots
$$
and which also counts rooted trivalent maps (OEIS \href{https://oeis.org/A062980}{A062980}; cf.~\cite{vidal2010,bodini-et-al}).
Besides the notion of $\alpha$-equivalence itself, the real interest of lambda calculus is that one can \emph{calculate} with it using the rules of \emph{$\beta$-reduction} and/or \emph{$\eta$-expansion}:
$$(\lambda x[t])(u) \overset\beta\to t\{u/x\}
\qquad
t \overset\eta\to \lambda x[t(x)]$$
Since $\beta$-reduction and $\eta$-expansion preserve linearity, the fragment of lambda calculus consisting of the linear lambda terms can be seen as a proper subsystem with various special properties (for example, computing the $\beta$-normal form of a linear lambda term is PTIME-complete \cite{mairson2004}, whereas for non-linear terms it is undecidable whether a normal form even exists).
Although it is sufficient to consider $\alpha$-equivalence for the purpose of presenting the bijection between linear lambda terms and rooted trivalent maps (which we will describe in \Cref{sec:linear-trivalent,sec:trivalent-linear}), we should nonetheless be aware that $\beta$-reduction and $\eta$-expansion are lurking in the background (and as mentioned in the introduction, there are some known connections between enumeration of $\beta$-normal terms and enumeration of rooted maps \cite{zg2015rpmnpt,z2015counting}).

\section{String diagrams for reflexive objects}
\label{sec:stringdiagrams}

\tikzset{->-/.style={decoration={
  markings,
  mark=at position .5 with {\arrow{>}}},postaction={decorate}}}
\tikzset{!->-/.style={decoration={
  markings,
  mark=at position #1 with {\arrow{>}}},postaction={decorate}}}

A natural way of visualizing a lambda term is to begin by drawing a tree representing the underlying structure of applications and abstractions, and then add extra edges connecting each bound variable occurrence to its corresponding abstraction.
For example, for the term $\lambda x[\lambda y[x(\lambda z[y(z)])]]$
we begin with the tree on the left, adding links to obtain the diagram on the right:
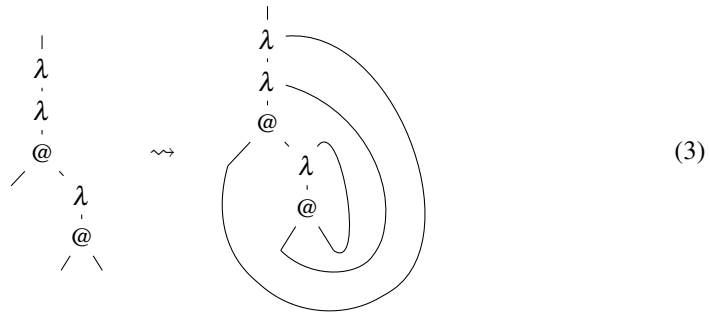
\begin{equation}\label{eqn:lambda-graph}
\vcenter{\hbox{
\begin{tikzpicture}[scale=0.8,level/.style={sibling distance=15em/#1, level distance=2em}]
  \node (z) {} child {
  node (a) {$\lambda$}
  child { node (b) {$\lambda$} 
    child {
      node (c) {$@$}
      child {node (d) {} }
      child {node (e) {$\lambda$}
        child {node (f) {$@$}
          child { node (g) {} }
          child { node (h) {} }
      }}
  }}};
\end{tikzpicture}}}
\quad\leadsto\quad
\vcenter{\hbox{
\begin{tikzpicture}[scale=0.8,level/.style={sibling distance=15em/#1, level distance=2em}]
  \node (z) {} child {
  node (a) {$\lambda$}
  child { node (b) {$\lambda$} 
    child {
      node (c) {$@$}
      child {node (d.center) {} }
      child {node (e) {$\lambda$}
        child {node (f) {$@$}
          child { node (g.center) {} }
          child { node (h.center) {} }
      }}
  }}};
  \node (x) [right=1em of h] {};
  \node (y) [below=8.5em of a] {};
  \node (z) [right=4em of y] {};
  \draw (d.center) to [bend right=30,rounded corners=5pt] (y) to [bend right=30,rounded corners=5pt] (z) to [bend right=75,rounded corners=5pt] (a);
  \draw[bend right=135] (h.center) to (e);
  \draw[bend right=45,rounded corners=5pt] (g.center) to (x.center) to (b);
\end{tikzpicture}}}
\end{equation}
This approach is especially natural for linear lambda terms, since each $\lambda$-abstraction binds exactly one variable occurrence.
Mairson \shortcite{mairson2002dilbert} refers to these kinds of diagrams as \emph{proof-nets} (probably because they are essentially equivalent to Girard's proof-nets for the implicative fragment of linear logic), while Bodini et.~al \shortcite{bodini-et-al} speak of them as ``syntactic trees''.
I do not know just how far back the idea goes, but I've even seen such a diagram used to display a (linear) lambda term in an old essay by Knuth \shortcite{knuth1970}, who called it a particular kind of \emph{information structure}.

In \Cref{sec:linear-trivalent,sec:trivalent-linear} we will use these types of diagrams to help explain the bijection between linear lambda terms and rooted trivalent maps.
The aim of this section is to briefly present a \emph{rational reconstruction} of the diagrams for the interested reader, using the framework of Joyal and Street \shortcite{joyal-street-i} to read these ``string diagrams'' as a notation for endomorphisms of a reflexive object.
I will assume that the reader already has some background in category theory -- for others, this section can be safely skipped, since the rest of the paper will only make use of the informal description of the diagrams, and not their categorical characterization.\footnote{The analysis given here follows \cite[Zeilberger and Giorgetti (2015, \S3.1)]{zg2015rpmnpt} but is considerably simplified; otherwise, the observation that the framework of string diagrams can be used to link lambda calculus proof-nets to reflexive objects is original as far as I know.}

An important insight of Dana Scott was that the equational theory of pure lambda calculus can be modelled using a \emph{reflexive object} in a cartesian closed category, in the sense of an object $U$ equipped with a retraction to its space of internal endomorphisms $U^U$ \cite{scott1980,hyland-lambda-calculus}.
To capture linear lambda calculus rather than classical lambda calculus, it suffices to rephrase Scott's original definition in the setting of \emph{symmetric monoidal closed} (smc) categories, replacing the exponential object $U^U$ by the internal hom object $\impR[U]{U}$ \cite{jacobs1993}.
Moreover, one can model the theory of ($\beta\eta$-)rewriting rather than the theory of equality by working with \emph{bicategories} rather than categories \cite{seely1987}.
Consider then the following definition:
\begin{definition}\label{defn:reflexive-object}
A \definand{reflexive object} in a smc bicategory $\mathcal{K}$ is an object $U$ equipped with an adjunction to its space of internal endomorphisms $\impR[U]{U}$.\footnote{This definition of reflexive object can be seen as a common generalization of two definitions introduced by Seely~\shortcite{seely1987} and by Jacobs~\shortcite{jacobs1993}, who considered, respectively, how to generalize Scott's definition to the setting of (cartesian closed) bicategories and of smc (1-)categories.}
\end{definition}
\noindent
A reflexive object in this sense consists of a pair of 1-cells
$$
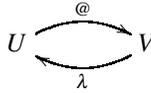
\xymatrixcolsep{3pc}
\xymatrixrowsep{2.5pc}
\xymatrix{U \rtwocell^@_\lambda{'} & V}
$$
where the object $V = \impR[U]{U}$ comes with an equivalence of categories $\mathcal{K}(X \otimes U,U) \cong \mathcal{K}(X,V)$ natural in $X$, 
together with a pair of 2-cells 
$$
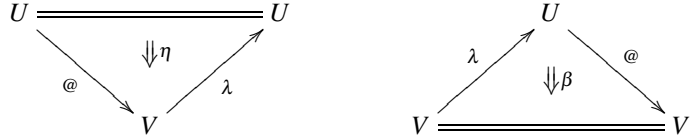
\xymatrixcolsep{3pc}
\xymatrixrowsep{2.5pc}
\xymatrix{\rrtwocell<\omit>{<3>\eta}U \ar[dr]_{@}\ar@{=}[rr] & & U \\  & V\ar[ur]_\lambda &}
\qquad\qquad
\xymatrix{\rrtwocell<\omit>{<5>\beta} & U\ar[dr]^{@} & \\ V \ar[ur]^{\lambda}\ar@{=}[rr] & & V}
$$
satisfying the zig-zag identities
$(\lambda \beta) \circ (\eta \lambda) = 1_\lambda$ and 
$(\beta @) \circ (@ \eta) = 1_@$.
Such data provides a model of linear lambda calculus in the following sense:
\begin{claim}[Soundness] \label{claim:soundness}
Let $U$ be a reflexive object in a smc bicategory $\mathcal{K}$.
Any linear lambda term $(\Gamma,t)$ can be interpreted as a 1-cell $\sem{t} : U^{\otimes k} \to U$ in $\mathcal{K}$, where $\Gamma = x_1,\dots,x_k$ and $U^{\otimes k}$ is the $k$-fold tensor product of $U$.
Moreover, this interpretation respects $\alpha$-equivalence, $\beta$-reduction, and $\eta$-expansion.
\end{claim}
\noindent
To make the connection between reflexive objects and the ``proof-nets'' for lambda terms, let's begin by observing that a special kind of smc bicategory is a \emph{compact closed} bicategory \cite{stay2013}, where the internal hom can be defined in terms of the tensor and dualization: $\impR[U]{U} \cong U \otimes U^*$.
There is a fairly standard set of conventions for drawing morphisms of compact closed (bi)categories as string diagrams \cite{selinger-survey,stay2013}, typically using orientations on the ``wires'' to distinguish between an object $A$ and its dual $A^*$.
Applying these conventions to the data of a reflexive object in a compact closed bicategory, the 1-cells $@ : U \to U \otimes U^*$ and $\lambda : U \otimes U^* \to U$ get drawn (running down the page) as nodes of the shape
$$
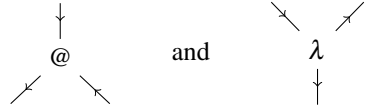
\vcenter{\hbox{\begin{tikzpicture}
  \node (fn) {};
  \node (app) [below=1em of fn] {$@$};
  \node (cont) [below left=1em of app] {};
  \node (arg) [below right=1em of app] {};
  \draw[->-] (fn.center) to (app);
  \draw[->-] (app) to (cont.center);  
  \draw[->-] (arg.center) to (app);  
\end{tikzpicture}}}
\qquad\text{and}\qquad
\vcenter{\hbox{\begin{tikzpicture}
  \node (root) {};
  \node (lam) [above=1em of root] {$\lambda$};
  \node (var) [above right=1em of lam] {};
  \node (body) [above left=1em of lam] {};
  \draw[->-] (lam) to (root.center);
  \draw[->-] (lam) to (var.center);  
  \draw[->-] (body.center) to (lam);  
\end{tikzpicture}}}
$$
while the 2-cells $\eta$ and $\beta$ become rewriting rules
$$
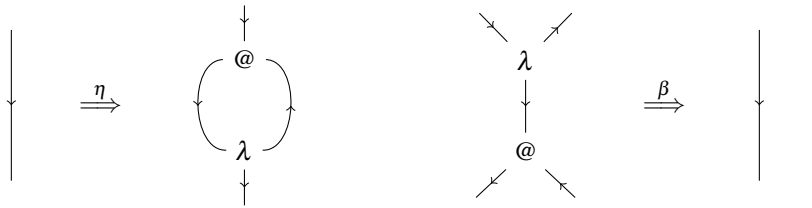
\vcenter{\hbox{
\begin{tikzpicture}
  \node (in) {};
  \node (out) [below=5em of in] {};
  \draw[->-] (in.center) to (out.center);
\end{tikzpicture}}}
\qquad\overset\eta\Longrightarrow\qquad
\vcenter{\hbox{
\begin{tikzpicture}
  \node (fn) {};
  \node (app) [below=1em of fn] {$@$};
  \node (lam) [below=2em of app] {$\lambda$};
  \node (root) [below=1em of lam]{};
  \draw[->-] (fn.center) to (app);
  \draw[->-,bend right=90] (app) to (lam);  
  \draw[->-,bend right=90] (lam) to (app);  
  \draw[->-] (lam) to (root.center);
\end{tikzpicture}}}
\qquad\qquad\qquad
\vcenter{\hbox{
\begin{tikzpicture}
  \node (lam) {$\lambda$};
  \node (var) [above right=1em of lam] {};
  \node (body) [above left=1em of lam] {};
  \node (app) [below=2em of lam] {$@$};
  \node (cont) [below left=1em of app] {};
  \node (arg) [below right=1em of app] {};
  \draw[->-] (body.center) to (lam);
  \draw[->-] (lam) to (var.center);
  \draw[->-] (lam) to (app);  
  \draw[->-] (app) to (cont.center);  
  \draw[->-] (arg.center) to (app);
\end{tikzpicture}}}
\qquad\overset\beta\Longrightarrow\qquad
\vcenter{\hbox{
\begin{tikzpicture}
  \node (in1) {};
  \node (in2) [right=1.5em of in1] {};
  \node (out1) [below=5em of in1] {};
  \node (out2) [below=5em of in2] {};
  \draw[->-] (in1.center) to (out1.center);
  \draw[->-] (out2.center) to (in2.center);
\end{tikzpicture}}}
$$
with the zig-zag identities expressing a coherence condition on these rewriting rules.
For ease of reference, we'll also give names (adopted from Mairson \shortcite{mairson2002dilbert}) to the different wires positioned around the 1-cells: running clockwise around an $@$-node, the incoming wire at the top is called the \definand{function} port, followed by the (incoming) \definand{argument} and (outgoing) \definand{continuation}, and running counterclockwise around a $\lambda$-node, the outgoing wire at the bottom is called the \definand{root} port, followed by the (outgoing) \definand{parameter} and (incoming) \definand{body}.\footnote{The reader might suspect that there is some degree of arbitrariness in these layout conventions, for example if we had used left implication $\impL[U]{U}$ instead of right implication $\impR[U]{U}$ in \Cref{defn:reflexive-object}.
The bijection we present in \Cref{sec:linear-trivalent,sec:trivalent-linear} works for any layout convention, and will only rely on having a particular convention fixed.
I should also point out this way of ordering the wires corresponds to what was briefly discussed as the ``RL'' convention in \cite[Zeilberger and Giorgetti (2015, \S3.1)]{zg2015rpmnpt}, rather than the ``LR'' convention which was mainly used in that paper.}

As an example, the closed linear lambda term $t = \lambda x[\lambda y[x(\lambda z[y(z)])]]$ we considered above denotes a morphism $\sem{t} : 1 \to U$ in any smc bicategory with a reflexive object $U$.
Drawing this 1-cell as a string diagram using the compact closed conventions we get the following picture:
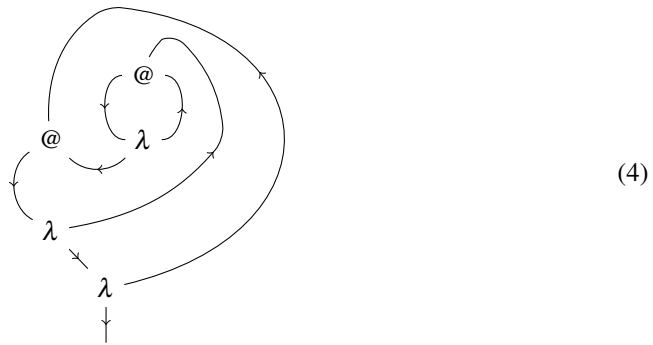
\begin{equation}\label{eqn:lambda-string-diagram}
\vcenter{\hbox{
\begin{tikzpicture}
  \node (root) {};
  \node (lamx) [above=1em of root] {$\lambda$};
  \node (lamy) [above left=1em of lamx] {$\lambda$};
  \node (appx) [above=2em of lamy] {$@$};
  \node (lamz) [right=2em of appx] {$\lambda$};
  \node (appy) [above=1em of lamz] {$@$};
  \node (vary) [above right=0.5em and 0em of appy] {};
  \node (dumy) [right=2em of lamz] {};
  \node (varx) [above right=4em and 1em of appx] {};
  \draw[->-] (lamx) to (root.center);
  \draw[->-] (lamy) to (lamx);
  \draw[->-,bend right=60] (appx) to (lamy);
  \draw[->-,bend left=45] (lamz) to (appx);
  \draw[->-,bend right=20,rounded corners=5pt] (lamy) to (dumy.center) to (vary.center) to (appy);
  \draw[->-,bend right=80,min distance=3cm,rounded corners=5pt] (lamx) to (varx.center) to [bend right=30,min distance=0em] (appx);
  \draw[->-,bend right=90] (appy) to (lamz);  
  \draw[->-,bend right=90] (lamz) to (appy);  
\end{tikzpicture}}}
\end{equation}
Observe that this diagram is essentially the same as the one in (\ref{eqn:lambda-graph}), just turned upside down and with explicit orientations on the wires.
The correspondence with the original linear lambda term can be made a bit more evident by labelling the wires with \emph{subterms} of $t$:
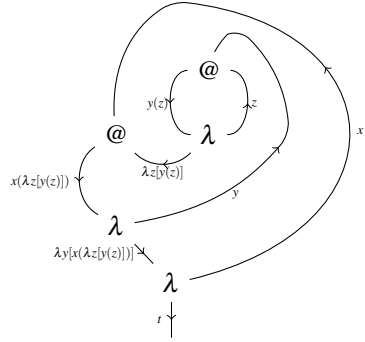
$$
\begin{tikzpicture}
  \node (root) {};
  \node (lamx) [above=1em of root] {$\lambda$};
  \node (lamy) [above left=1em of lamx] {$\lambda$};
  \node (appx) [above=2em of lamy] {$@$};
  \node (lamz) [right=2em of appx] {$\lambda$};
  \node (appy) [above=1em of lamz] {$@$};
  \node (vary) [above right=0.5em and 0em of appy] {};
  \node (dumy) [right=2em of lamz] {};
  \node (varx) [above right=4em and 1em of appx] {};
  \draw[->-] (lamx) to node [left] {\tiny$t$}  (root.center);
  \draw[->-] (lamy) to node [left] {\tiny$\lambda y[x(\lambda z[y(z)])]$} (lamx);
  \draw[->-,bend right=60] (appx) to node [left] {\tiny$x(\lambda z[y(z)])$}  (lamy);
  \draw[->-,bend left=45] (lamz) to node [below=-3pt] {\tiny$\lambda z[y(z)]$} (appx);
  \draw[->-,bend right=20,rounded corners=5pt] (lamy) to node [right=1pt] {\tiny$y$} (dumy.center) to (vary.center) to (appy);
  \draw[->-,bend right=80,min distance=3cm,rounded corners=5pt] (lamx) to node [right] {\tiny$x$} (varx.center) to [bend right=30,min distance=0em] (appx);
  \draw[->-,bend right=90] (appy) to node [left=-3pt] {\tiny$y(z)$} (lamz);  
  \draw[->-,bend right=90] (lamz) to node [right=-2pt] {\tiny$z$} (appy);  
\end{tikzpicture}
$$
One thing it is important to point out is that not every physical combination of $@$-nodes and $\lambda$-nodes represents a linear lambda term, a consequence of the fact that not every smc (bi)category is compact closed.
For instance, the diagrams
$$
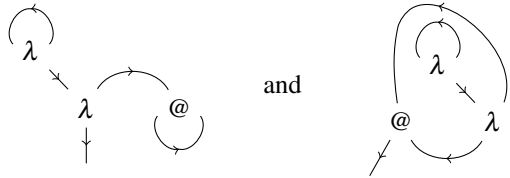
\vcenter{\hbox{
\begin{tikzpicture}
  \node (root) {};
  \node (lam) [above=1em of root] {$\lambda$};
  \node (lam2) [above left=1em of lam] {$\lambda$};
  \node (app) [right=2em of lam] {$@$};
  \node (xl) [above=0.5em of lam2] {};
  \node (xr) [below=0.5em of app] {};
  \draw[->-] (lam) to (root.center);
  \draw[->-,bend left=60] (lam) to (app);  
  \draw[->-] (lam2) to (lam);  
  \draw[->-,bend right=70,looseness=1] (lam2) to (xl.center) to (lam2);  
  \draw[->-,bend right=70,looseness=1] (app) to (xr.center) to (app);  
\end{tikzpicture}}}
\qquad\text{and}\qquad
\vcenter{\hbox{
\begin{tikzpicture}
  \node (root) {};
  \node (appx) [above right=1em and 0em of root] {$@$};
  \node (lamx) [right=2em of appx] {$\lambda$};
  \node (lamy) [above left=1em of lamx] {$\lambda$};
  \node (vary) [above=0.5em of lamy] {};
  \node (varx) [above=3em of appx] {};
  \draw[->-] (appx) to (root.center);
  \draw[->-,bend left=60] (lamx) to (appx);  
  \draw[->-] (lamy) to (lamx);  
  \draw[->-,bend right=70,looseness=1] (lamy) to (vary.center) to (lamy);  
  \draw[->-,bend right=70,looseness=1,rounded corners=5pt] (lamx) to (varx.center) to [bend right=10] (appx);  
\end{tikzpicture}}}
$$
do not correspond to the interpretation of a linear lambda term.
(This phenomenon is well-known in proof-nets, and is often analyzed by considering additional ``correctness criteria'' for the diagrams.)
Nonetheless, the interpretation of linear lambda terms using reflexive objects in smc bicategories is complete in the following sense:
\begin{claim}[Completeness] \label{claim:completenes}
  There is a smc bicategory $\mathcal{K}_\Lambda$ equipped with a reflexive object $U$, such that every 1-cell $f : U^{\otimes k} \to U$ is the interpretation $f = \sem{t}$ of a unique (up to $\alpha$-equivalence) linear lambda term $t$ with $k$ free variables, and such that there is a 2-cell $\sem{t_1} \Rightarrow \sem{t_2}$ if and only if $t_1$ can be rewritten to $t_2$ (up to $\alpha$-equivalence) by a series of $\beta$-reductions and $\eta$-expansions.
\end{claim}
\noindent
The proof essentially follows Hyland's analysis of Scott's Representation Theorem \cite{hyland-lambda-calculus}, replacing cartesian closed categories by smc bicategories.
The idea is to take $\mathcal{K}_\Lambda$ as a presheaf bicategory $[\mathcal{C}^\op,\Cat]$, where $\mathcal{C}$ is a symmetric monoidal bicategory whose 0-cells are contexts, 1-cells are tuples of linear lambda terms, and 2-cells are rewritings between tuples.
The smc structure on $[\mathcal{C}^\op,\Cat]$ is defined by Day convolution, and the reflexive object is constructed as the representable presheaf for a singleton context.
(Note that $[\mathcal{C}^\op,\Cat]$ is \emph{not} compact closed.)

\section{From linear lambda terms to rooted trivalent maps}
\label{sec:linear-trivalent}

Once we view linear lambda calculus through the lens of string diagrams, it is pretty clear how to turn any closed linear lambda term into a rooted trivalent map: just look at its string diagram and forget the distinction between $@$-nodes and $\lambda$-nodes, as well as the orientations on the wires.
Here we apply this transformation on $\lambda x[\lambda y[x(\lambda z[y(z)])]]$ and its corresponding string diagram (\ref{eqn:lambda-string-diagram}):
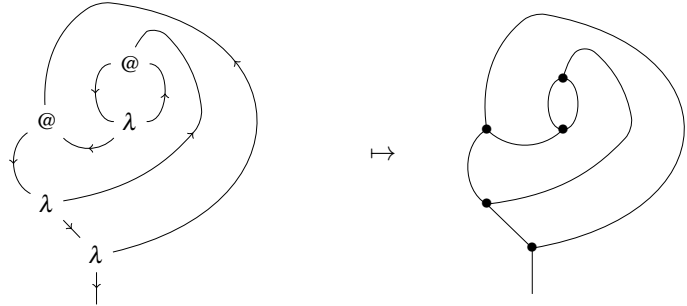
$$
\vcenter{\hbox{\scalebox{0.9}{
\begin{tikzpicture}
  \node (root) {};
  \node (lamx) [above=1em of root] {$\lambda$};
  \node (lamy) [above left=1em of lamx] {$\lambda$};
  \node (appx) [above=2em of lamy] {$@$};
  \node (lamz) [right=2em of appx] {$\lambda$};
  \node (appy) [above=1em of lamz] {$@$};
  \node (vary) [above right=0.5em and 0em of appy] {};
  \node (dumy) [right=2em of lamz] {};
  \node (varx) [above right=4em and 1em of appx] {};
  \draw[->-] (lamx) to (root.center);
  \draw[->-] (lamy) to (lamx);
  \draw[->-,bend right=60] (appx) to (lamy);
  \draw[->-,bend left=45] (lamz) to (appx);
  \draw[->-,bend right=20,rounded corners=5pt] (lamy) to (dumy.center) to (vary.center) to (appy);
  \draw[->-,bend right=80,min distance=3cm,rounded corners=5pt] (lamx) to (varx.center) to [bend right=30,min distance=0em] (appx);
  \draw[->-,bend right=90] (appy) to (lamz);  
  \draw[->-,bend right=90] (lamz) to (appy);  
\end{tikzpicture}}}}
\qquad\mapsto\qquad
\vcenter{\hbox{\scalebox{0.9}{
\begin{tikzpicture}
  \node (root) {};
  \node (lamx) [above=1em of root] {$\bullet$};
  \node (lamy) [above left=1em of lamx] {$\bullet$};
  \node (appx) [above=2em of lamy] {$\bullet$};
  \node (lamz) [right=2em of appx] {$\bullet$};
  \node (appy) [above=1em of lamz] {$\bullet$};
  \node (vary) [above right=0.5em and 0em of appy] {};
  \node (dumy) [right=2em of lamz] {};
  \node (varx) [above right=4em and 1em of appx] {};
  \draw (lamx.center) to (root.center);
  \draw (lamy.center) to (lamx.center);
  \draw[bend right=60] (appx.center) to (lamy.center);
  \draw[bend left=45] (lamz.center) to (appx.center);
  \draw[bend right=20,rounded corners=5pt] (lamy.center) to (dumy.center) to (vary.center) to (appy.center);
  \draw[bend right=80,min distance=3cm,rounded corners=5pt] (lamx.center) to (varx.center) to [bend right=30,min distance=0em] (appx.center);
  \draw[bend right=90] (appy.center) to (lamz.center);  
  \draw[bend right=90] (lamz.center) to (appy.center);  
\end{tikzpicture}}}}
$$
Even though we identify $@$-nodes and $\lambda$-nodes, it is important that we take care to remember the \emph{ordering} of the wires around each node, since we are interested in obtaining a combinatorial map rather than an abstract trivalent graph.
Strictly speaking the diagram on the right is not a trivalent map in the sense of \Cref{def:trivalent} since it has a dangling edge, but it can be interpreted as a rooted trivalent map (in the sense of \Cref{def:rooted}) by reading the outgoing trivalent vertex as a ``normal vector'' to the root dart of a map with that vertex smoothed out:
$$
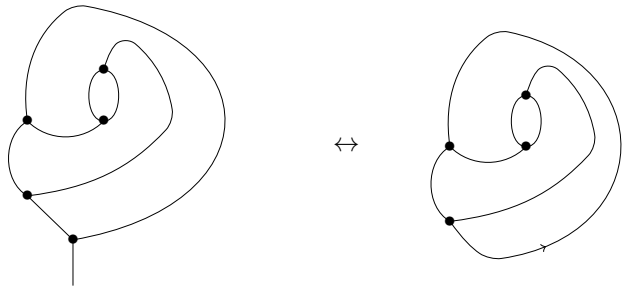
\vcenter{\hbox{\scalebox{0.9}{
\begin{tikzpicture}
  \node (root) {};
  \node (lamx) [above=1em of root] {$\bullet$};
  \node (lamy) [above left=1em of lamx] {$\bullet$};
  \node (appx) [above=2em of lamy] {$\bullet$};
  \node (lamz) [right=2em of appx] {$\bullet$};
  \node (appy) [above=1em of lamz] {$\bullet$};
  \node (vary) [above right=0.5em and 0em of appy] {};
  \node (dumy) [right=2em of lamz] {};
  \node (varx) [above right=4em and 1em of appx] {};
  \draw (lamx.center) to (root.center);
  \draw (lamy.center) to (lamx.center);
  \draw[bend right=60] (appx.center) to (lamy.center);
  \draw[bend left=45] (lamz.center) to (appx.center);
  \draw[bend right=20,rounded corners=5pt] (lamy.center) to (dumy.center) to (vary.center) to (appy.center);
  \draw[bend right=80,min distance=3cm,rounded corners=5pt] (lamx.center) to (varx.center) to [bend right=30,min distance=0em] (appx.center);
  \draw[bend right=90] (appy.center) to (lamz.center);  
  \draw[bend right=90] (lamz.center) to (appy.center);  
\end{tikzpicture}}}}
\qquad\leftrightarrow\qquad
\vcenter{\hbox{\scalebox{0.9}{
\begin{tikzpicture}
  \node (lamx-inv) {};
  \node (lamy) [above left=1em of lamx-inv] {$\bullet$};
  \node (appx) [above=2em of lamy] {$\bullet$};
  \node (lamz) [right=2em of appx] {$\bullet$};
  \node (appy) [above=1em of lamz] {$\bullet$};
  \node (vary) [above right=0.5em and 0em of appy] {};
  \node (dumy) [right=2em of lamz] {};
  \node (varx) [above right=4em and 1em of appx] {};
  \draw[bend right=60] (appx.center) to (lamy.center);
  \draw[bend left=45] (lamz.center) to (appx.center);
  \draw[bend right=20,rounded corners=5pt] (lamy.center) to (dumy.center) to (vary.center) to (appy.center);
  \draw[!->-=.2,bend right=10,rounded corners=5pt] (lamy.center) to (lamx-inv.center) to [bend right=80,min distance=2.5cm,rounded corners=5pt] (varx.center) to [bend right=30,min distance=0em] (appx.center);
  \draw[bend right=90] (appy.center) to (lamz.center);  
  \draw[bend right=90] (lamz.center) to (appy.center);  
\end{tikzpicture}}}}
$$
Actually, there is a hiccup in performing this last step for the identity term $\Icomb = \lambda x[x]$,
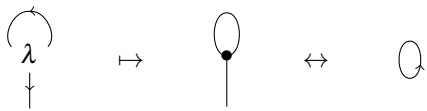
$$
\vcenter{\hbox{
\begin{tikzpicture}
  \node (root) {};
  \node (lamx) [above=1em of root] {$\lambda$};
  \node (varx) [above=0.5em of lamx] {};
  \draw[->-] (lamx) to (root.center);
  \draw[->-,bend right=70,looseness=1] (lamx) to (varx.center) to (lamx);  
\end{tikzpicture}}}
\qquad\mapsto\qquad
\vcenter{\hbox{
\begin{tikzpicture}
  \node (root) {};
  \node (lamx) [above=1em of root] {$\bullet$};
  \node (varx) [above=0.7em of lamx] {};
  \draw (lamx.center) to (root.center);
  \draw[bend right=90,looseness=1] (lamx.center) to (varx.center) to (lamx.center);  
\end{tikzpicture}}}
\qquad\leftrightarrow\qquad
\vcenter{\hbox{
\begin{tikzpicture}
  \node (lamx-inv) {};
  \node (varx) [above=0.7em of lamx-inv] {};
  \draw[!->-=0.2,bend right=90,looseness=1] (lamx-inv.center) to (varx.center) to (lamx-inv.center);  
\end{tikzpicture}}}
$$
since the no-vertex map is not technically a rooted map (again in the sense of \Cref{def:rooted}), although studies of the combinatorics of rooted maps often treat the empty map as an exceptional case \cite{tutte1968}.

Both of these minor technical issues will be resolved smoothly once we adopt the more general notion of rooted trivalent map to be described shortly.
The real reason for considering this more general notion, though, is that we would also like to interpret linear lambda terms with free variables as rooted trivalent maps.
Consider the term $x(\lambda z[y(z)])$ with free variables $x$ and $y$, whose string diagram corresponds to a subdiagram of (\ref{eqn:lambda-string-diagram}):
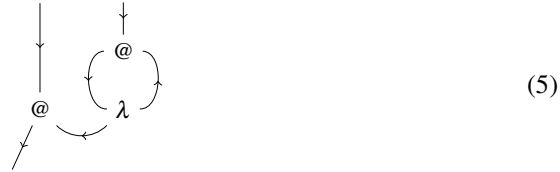
\begin{equation}\label{eqn:subdiagram}
\vcenter{\hbox{\scalebox{0.9}{
\begin{tikzpicture}
  \node (root) {};
  \node (appx) [above right=1.5em and 0em of root] {$@$};
  \node (lamz) [right=2em of appx] {$\lambda$};
  \node (appy) [above=1em of lamz] {$@$};
  \node (vary) [above=1em of appy] {};
  \node (varx) [above=3.4em of appx] {};
  \draw[->-] (appx) to (root.center);
  \draw[->-,bend left=45] (lamz) to (appx);
  \draw[->-] (vary.center) to (appy);
  \draw[->-] (varx.center) to (appx);
  \draw[->-,bend right=90] (appy) to (lamz);  
  \draw[->-,bend right=90] (lamz) to (appy);  
\end{tikzpicture}}}}
\end{equation}
Identifying $@$-nodes and $\lambda$-nodes in (\ref{eqn:subdiagram}) yields a trivalent map with three dangling edges,
$$
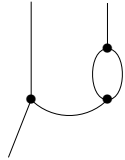
\vcenter{\hbox{\scalebox{0.9}{
\begin{tikzpicture}
  \node (root) {};
  \node (appx) [above right=1.5em and 0em of root] {$\bullet$};
  \node (lamz) [right=2em of appx] {$\bullet$};
  \node (appy) [above=1em of lamz] {$\bullet$};
  \node (vary) [above=1em of appy] {};
  \node (varx) [above=3.2em of appx] {};
  \draw (appx.center) to (root.center);
  \draw[bend left=45] (lamz.center) to (appx.center);
  \draw (vary.center) to (appy.center);
  \draw (varx.center) to (appx.center);
  \draw[bend right=90] (appy.center) to (lamz.center);  
  \draw[bend right=90] (lamz.center) to (appy.center);  
\end{tikzpicture}}}}
$$
but since it is no longer clear from the diagram which dangling edge marks the root,
we attach an extra univalent vertex to one of them (turning it into a full edge):
$$
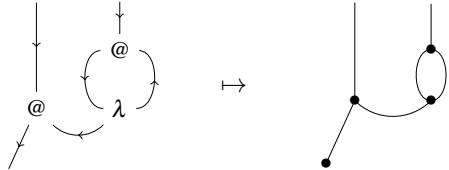
\vcenter{\hbox{\scalebox{0.9}{
\begin{tikzpicture}
  \node (root) {};
  \node (appx) [above right=1.5em and 0em of root] {$@$};
  \node (lamz) [right=2em of appx] {$\lambda$};
  \node (appy) [above=1em of lamz] {$@$};
  \node (vary) [above=1em of appy] {};
  \node (varx) [above=3.4em of appx] {};
  \draw[->-] (appx) to (root.center);
  \draw[->-,bend left=45] (lamz) to (appx);
  \draw[->-] (vary.center) to (appy);
  \draw[->-] (varx.center) to (appx);
  \draw[->-,bend right=90] (appy) to (lamz);  
  \draw[->-,bend right=90] (lamz) to (appy);  
\end{tikzpicture}}}}
\qquad\mapsto\qquad
\vcenter{\hbox{\scalebox{0.9}{
\begin{tikzpicture}
  \node (root) {$\bullet$};
  \node (appx) [above right=1.5em and 0em of root] {$\bullet$};
  \node (lamz) [right=2em of appx] {$\bullet$};
  \node (appy) [above=1em of lamz] {$\bullet$};
  \node (vary) [above=1em of appy] {};
  \node (varx) [above=3.2em of appx] {};
  \draw (appx.center) to (root.center);
  \draw[bend left=45] (lamz.center) to (appx.center);
  \draw (vary.center) to (appy.center);
  \draw (varx.center) to (appx.center);
  \draw[bend right=90] (appy.center) to (lamz.center);  
  \draw[bend right=90] (lamz.center) to (appy.center);  
\end{tikzpicture}}}}
$$
By considering the output of this transformation on linear lambda terms with any number of free variables we arrive at the following generalization of \Cref{def:rooted}:
\begin{definition}\label{def:rtmb}
  For any $\gset{G}$ $X$ and $g\in G$, let $\fixpt{g}(X) \defeq \set{x \in X \mid g * x = x}$, and again take $\trigp \defeq \freegp{v,e \mid v^3 = e^2 = 1}$.
  A \definand{rooted trivalent map with boundary} is a transitive $\gset\trigp$ $M$ with a distinguished element $r\in M$ and a list of distinct elements $x_1,\dots,x_k \in M$, such that $\fixpt{v}(M) = \set{r}$ and $\fixpt{e}(M) = \set{x_1,\dots,x_k}$.
  We refer to the unique $v$-fixed point as the \definand{root} $r(M)$ of the map, the ordered list of $e$-fixed points as the \definand{boundary} $\Gamma(M)$ of the map, and to the integer $k$ (= the number of $e$-fixed points) as the \definand{degree} of the boundary.
  A rooted trivalent map with boundary of degree 0 is called a \definand{closed} rooted trivalent map.
\end{definition}
\noindent
From now on, when we say ``rooted trivalent map'' without qualification we mean rooted trivalent map with boundary in the sense of \Cref{def:rtmb}, referring to the sense of \Cref{def:rooted} as ``classical rooted trivalent map''.
\begin{proposition}\label{prop:bijroot}
  For all $n>0$, there is a bijection between closed rooted trivalent maps with $n+1$ trivalent vertices and classical rooted trivalent maps with $n$ trivalent vertices.
  This extends to a bijection for all $n \ge 0$ if the empty $\gset\trigp$ is admitted as a classical rooted trivalent map.
\end{proposition}
\begin{proof}
  As explained in the first paragraph of this section.
\end{proof}
\noindent
Observe that the simplest possible rooted trivalent map with boundary is the singleton $\gset{\trigp}$ $M = \set{x}$ with $r(M) = \Gamma(M) = x$, corresponding to the trivial map
$\vcenter{\hbox{\scalebox{0.6}{
\begin{tikzpicture}
  \node (root) {$\bullet$};
  \node (varx) [above=0.5em of root] {};
  \draw (varx.center) to (root.center);
\end{tikzpicture}}}}$
with no trivalent vertices and one free edge.
With this definition, it is clear how any linear lambda term induces a rooted trivalent map with boundary.
\begin{proposition}\label{prop:linear-trivalent}
  To any linear lambda term with $k$ free variables, $p$ applications and $q$ abstractions there is naturally associated a rooted trivalent map with boundary of degree $k$ and $p+q$ trivalent vertices.
\end{proposition}
\begin{proof}
Consider the string diagram of the term, which has $k$ incoming wires and one outgoing wire, as well as $p$ $@$-nodes and $q$ $\lambda$-nodes internal to the diagram.
Transform $@$-nodes and $\lambda$-nodes into trivalent vertices, attach a univalent vertex to the end of the outgoing wire, and finally forget the orientations of the wires.
The result is manifestly a rooted trivalent map with boundary of degree $k$ and $p+q$ trivalent vertices.
\end{proof}
\noindent
We call the map described in the proof of \Cref{prop:linear-trivalent} the \definand{underlying rooted trivalent map} of a linear lambda term.
Although the high-level description is clear enough, we can also explicitly compute the permutations $v$ and $e$ associated to the underlying rooted trivalent map by induction on the structure of the linear lambda term, which just requires a bit of bookkeeping.
What is somewhat more surprising is that this ``forgetful'' transformation can be reversed -- our next topic.

\section{From rooted trivalent maps to linear lambda terms}
\label{sec:trivalent-linear}

Given a rooted trivalent map $M$ (with a boundary of dangling edges $\Gamma(M)$), our task is to find a linear lambda term (with free variables $\Gamma$) whose underlying rooted trivalent map is $M$ (hopefully, such a linear lambda term always exists and is unique!).
To be able to do this, clearly we will somehow have to decide for each trivalent vertex whether it corresponds to an application ($@$) or an abstraction ($\lambda$).
The trick is that there is always an immediate answer for the trivalent vertex incident to the root: if removing that vertex disconnects the underlying graph then the vertex corresponds to an $@$-node, and otherwise it corresponds to a $\lambda$-node.
In either case, we can re-root the resulting submap(s) (while adjusting the boundary)
$$
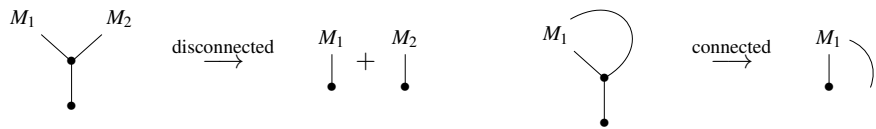
\vcenter{\hbox{\scalebox{0.8}{
\begin{tikzpicture}
  \node (root) {$\bullet$};
  \node (app) [above=1em of root] {$\bullet$};
  \node (map1) [above left=1em of app] {$M_1$};
  \node (map2) [above right=1em of app] {$M_2$};
  \draw (app.center) to (root.center);
  \draw (map1) to (app.center);
  \draw (map2) to (app.center);
\end{tikzpicture}}}}
\quad\overset{\text{disconnected}}\longrightarrow\quad
\vcenter{\hbox{\scalebox{0.8}{
\begin{tikzpicture}
  \node (root) {$\bullet$};
  \node (map1) [above=1em of root] {$M_1$};
  \draw (map1) to (root.center);
\end{tikzpicture}}}}
+ 
\vcenter{\hbox{\scalebox{0.8}{
\begin{tikzpicture}
  \node (root) {$\bullet$};
  \node (map2) [above=1em of root] {$M_2$};
  \draw (map2) to (root.center);
\end{tikzpicture}}}}
\qquad\qquad
\vcenter{\hbox{\scalebox{0.8}{
\begin{tikzpicture}
  \node (root) {$\bullet$};
  \node (lam) [above=1em of root] {$\bullet$};
  \node (map1) [above left=1em of lam] {$M_1$};
  \draw (lam.center) to (root.center);
  \draw (map1) to (lam.center);
  \draw[bend right=90,min distance=3em] (lam.center) to (map1);
\end{tikzpicture}}}}
\quad\overset{\text{connected}}\longrightarrow\quad
\vcenter{\hbox{\scalebox{0.8}{
\begin{tikzpicture}
  \node (root) {$\bullet$};
  \node (map1) [above=1em of root] {$M_1$};
  \node (varx) [right=1em of root] {};
  \draw (map1) to (root.center);
  \draw [bend right=45,min distance=1em] (varx.center) to (map1);
\end{tikzpicture}}}}
$$
and continue the process recursively until we eventually arrive at the trivial map 
$\vcenter{\hbox{\scalebox{0.6}{
\begin{tikzpicture}
  \node (root) {$\bullet$};
  \node (varx) [above=0.5em of root] {};
  \draw (varx.center) to (root.center);
\end{tikzpicture}}}}$.
For example, here we start applying this process to a closed rooted trivalent map
$$
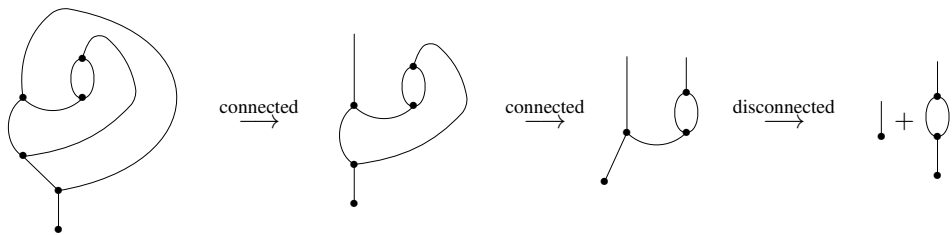
\vcenter{\hbox{\scalebox{0.7}{
\begin{tikzpicture}
  \node (root) {$\bullet$};
  \node (lamx) [above=1em of root] {$\bullet$};
  \node (lamy) [above left=1em of lamx] {$\bullet$};
  \node (appx) [above=2em of lamy] {$\bullet$};
  \node (lamz) [right=2em of appx] {$\bullet$};
  \node (appy) [above=1em of lamz] {$\bullet$};
  \node (vary) [above right=0.5em and 0em of appy] {};
  \node (dumy) [right=2em of lamz] {};
  \node (varx) [above right=4em and 1em of appx] {};
  \draw (lamx.center) to (root.center);
  \draw (lamy.center) to (lamx.center);
  \draw[bend right=60] (appx.center) to (lamy.center);
  \draw[bend left=45] (lamz.center) to (appx.center);
  \draw[bend right=20,rounded corners=5pt] (lamy.center) to (dumy.center) to (vary.center) to (appy.center);
  \draw[bend right=80,min distance=3cm,rounded corners=5pt] (lamx.center) to (varx.center) to [bend right=30,min distance=0em] (appx.center);
  \draw[bend right=90] (appy.center) to (lamz.center);  
  \draw[bend right=90] (lamz.center) to (appy.center);  
\end{tikzpicture}}}}
\overset{\text{connected}}\longrightarrow\quad
\vcenter{\hbox{\scalebox{0.7}{
\begin{tikzpicture}
  \node (root) {$\bullet$};
  \node (lamy) [above=1em of root] {$\bullet$};
  \node (appx) [above=2em of lamy] {$\bullet$};
  \node (lamz) [right=2em of appx] {$\bullet$};
  \node (appy) [above=1em of lamz] {$\bullet$};
  \node (vary) [above right=0.5em and 0em of appy] {};
  \node (dumy) [right=2em of lamz] {};
  \node (varx) [above=3em of appx] {};
  \draw (lamy.center) to (root.center);
  \draw[bend right=60] (appx.center) to (lamy.center);
  \draw[bend left=45] (lamz.center) to (appx.center);
  \draw[bend right=20,rounded corners=5pt] (lamy.center) to (dumy.center) to (vary.center) to (appy.center);
  \draw (varx.center) to (appx.center);
  \draw[bend right=90] (appy.center) to (lamz.center);  
  \draw[bend right=90] (lamz.center) to (appy.center);  
\end{tikzpicture}}}}
\quad\overset{\text{connected}}\longrightarrow
\vcenter{\hbox{\scalebox{0.7}{
\begin{tikzpicture}
  \node (root) {$\bullet$};
  \node (appx) [above right=1.5em and 0em of root] {$\bullet$};
  \node (lamz) [right=2em of appx] {$\bullet$};
  \node (appy) [above=1em of lamz] {$\bullet$};
  \node (vary) [above=1em of appy] {};
  \node (varx) [above=3.2em of appx] {};
  \draw (appx.center) to (root.center);
  \draw[bend left=45] (lamz.center) to (appx.center);
  \draw (vary.center) to (appy.center);
  \draw (varx.center) to (appx.center);
  \draw[bend right=90] (appy.center) to (lamz.center);  
  \draw[bend right=90] (lamz.center) to (appy.center);  
\end{tikzpicture}}}}
\quad\overset{\text{disconnected}}\longrightarrow\quad
\vcenter{\hbox{\scalebox{0.7}{
\begin{tikzpicture}
  \node (root) {$\bullet$};
  \node (varx) [above=1em of root] {};
  \draw (varx.center) to (root.center);
\end{tikzpicture}}}}
+ 
\vcenter{\hbox{\scalebox{0.7}{
\begin{tikzpicture}
  \node (root) {$\bullet$};
  \node (lamz) [above=1em of root] {$\bullet$};
  \node (appy) [above=1em of lamz] {$\bullet$};
  \node (vary) [above=1em of appy] {};
  \node (varz) [right=1em of vary] {};
  \draw (vary.center) to (appy.center);
  \draw[bend right=90] (appy.center) to (lamz.center);  
  \draw[bend right=90] (lamz.center) to (appy.center);  
  \draw (lamz.center) to (root.center);  
\end{tikzpicture}}}}
$$
and continue until we are left with only trivial maps:
$$
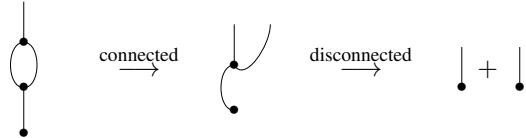
\vcenter{\hbox{\scalebox{0.8}{
\begin{tikzpicture}
  \node (root) {$\bullet$};
  \node (lamz) [above=1em of root] {$\bullet$};
  \node (appy) [above=1em of lamz] {$\bullet$};
  \node (vary) [above=1em of appy] {};
  \node (varz) [right=1em of vary] {};
  \draw (vary.center) to (appy.center);
  \draw[bend right=90] (appy.center) to (lamz.center);  
  \draw[bend right=90] (lamz.center) to (appy.center);  
  \draw (lamz.center) to (root.center);  
\end{tikzpicture}}}}
\quad\overset{\text{connected}}\longrightarrow\quad
\vcenter{\hbox{\scalebox{0.8}{
\begin{tikzpicture}
  \node (root) {$\bullet$};
  \node (appy) [above=1em of root] {$\bullet$};
  \node (vary) [above=1em of appy] {};
  \node (varz) [right=1em of vary] {};
  \draw (vary.center) to (appy.center);
  \draw[bend right=90] (appy.center) to (root.center);  
  \draw[out=-90,in=-60] (varz.center) to (appy.center);  
\end{tikzpicture}}}}
\quad\overset{\text{disconnected}}\longrightarrow\quad
\vcenter{\hbox{\scalebox{0.8}{
\begin{tikzpicture}
  \node (root) {$\bullet$};
  \node (varx) [above=1em of root] {};
  \draw (varx.center) to (root.center);
\end{tikzpicture}}}}
+
\vcenter{\hbox{\scalebox{0.8}{
\begin{tikzpicture}
  \node (root) {$\bullet$};
  \node (varx) [above=1em of root] {};
  \draw (varx.center) to (root.center);
\end{tikzpicture}}}}
$$
From the trace of this decomposition, we can reconstruct the necessarily unique linear lambda term whose underlying rooted trivalent map is our original map:
$$
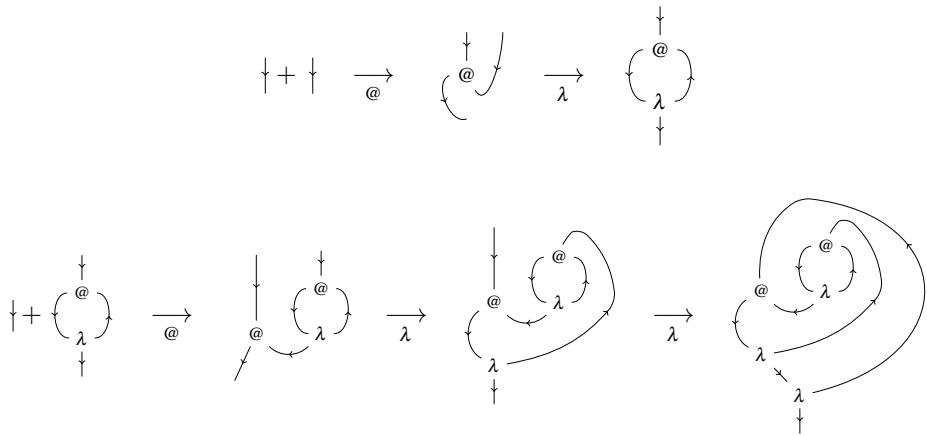
\vcenter{\hbox{\scalebox{0.8}{
\begin{tikzpicture}
  \node (root) {};
  \node (varx) [above=1em of root] {};
  \draw[->-] (varx.center) to (root.center);
\end{tikzpicture}}}}
+
\vcenter{\hbox{\scalebox{0.8}{
\begin{tikzpicture}
  \node (root) {};
  \node (varx) [above=1em of root] {};
  \draw[->-] (varx.center) to (root.center);
\end{tikzpicture}}}}
\quad\underset{@}\longrightarrow\quad
\vcenter{\hbox{\scalebox{0.8}{
\begin{tikzpicture}
  \node (root) {};
  \node (appy) [above=1em of root] {$@$};
  \node (vary) [above=1em of appy] {};
  \node (varz) [right=1em of vary] {};
  \draw[->-] (vary.center) to (appy);
  \draw[->-,bend right=90] (appy) to (root.center);  
  \draw[->-,out=-90,in=-60] (varz.center) to (appy);  
\end{tikzpicture}}}}
\quad\underset{\lambda}\longrightarrow\quad
\vcenter{\hbox{\scalebox{0.8}{
\begin{tikzpicture}
  \node (root) {};
  \node (lamz) [above=1em of root] {$\lambda$};
  \node (appy) [above=1em of lamz] {$@$};
  \node (vary) [above=1em of appy] {};
  \node (varz) [right=1em of vary] {};
  \draw[->-] (vary.center) to (appy);
  \draw[->-,bend right=90] (appy) to (lamz);  
  \draw[->-,bend right=90] (lamz) to (appy);  
  \draw[->-] (lamz) to (root.center);  
\end{tikzpicture}}}}
$$

$$
\vcenter{\hbox{\scalebox{0.7}{
\begin{tikzpicture}
  \node (root) {};
  \node (varx) [above=1em of root] {};
  \draw[->-] (varx.center) to (root.center);
\end{tikzpicture}}}}
+
\vcenter{\hbox{\scalebox{0.7}{
\begin{tikzpicture}
  \node (root) {};
  \node (lamz) [above=1em of root] {$\lambda$};
  \node (appy) [above=1em of lamz] {$@$};
  \node (vary) [above=1em of appy] {};
  \node (varz) [right=1em of vary] {};
  \draw[->-] (vary.center) to (appy);
  \draw[->-,bend right=90] (appy) to (lamz);  
  \draw[->-,bend right=90] (lamz) to (appy);  
  \draw[->-] (lamz) to (root.center);  
\end{tikzpicture}}}}
\quad\underset{@}\longrightarrow\quad
\vcenter{\hbox{\scalebox{0.7}{
\begin{tikzpicture}
  \node (root) {};
  \node (appx) [above right=1.5em and 0em of root] {$@$};
  \node (lamz) [right=2em of appx] {$\lambda$};
  \node (appy) [above=1em of lamz] {$@$};
  \node (vary) [above=1em of appy] {};
  \node (varx) [above=3.2em of appx] {};
  \draw[->-] (appx) to (root.center);
  \draw[->-,bend left=45] (lamz) to (appx);
  \draw[->-] (vary.center) to (appy);
  \draw[->-] (varx.center) to (appx);
  \draw[->-,bend right=90] (appy) to (lamz);  
  \draw[->-,bend right=90] (lamz) to (appy);  
\end{tikzpicture}}}}
\quad\underset{\lambda}\longrightarrow\quad
\vcenter{\hbox{\scalebox{0.7}{
\begin{tikzpicture}
  \node (root) {};
  \node (lamy) [above=1em of root] {$\lambda$};
  \node (appx) [above=2em of lamy] {$@$};
  \node (lamz) [right=2em of appx] {$\lambda$};
  \node (appy) [above=1em of lamz] {$@$};
  \node (vary) [above right=0.5em and 0em of appy] {};
  \node (dumy) [right=2em of lamz] {};
  \node (varx) [above=3em of appx] {};
  \draw[->-] (lamy) to (root.center);
  \draw[->-,bend right=60] (appx) to (lamy);
  \draw[->-,bend left=45] (lamz) to (appx);
  \draw[->-,bend right=20,rounded corners=5pt] (lamy) to (dumy.center) to (vary.center) to (appy);
  \draw[->-] (varx.center) to (appx);
  \draw[->-,bend right=90] (appy) to (lamz);  
  \draw[->-,bend right=90] (lamz) to (appy);  
\end{tikzpicture}}}}
\quad\underset{\lambda}\longrightarrow\quad
\vcenter{\hbox{\scalebox{0.7}{
\begin{tikzpicture}
  \node (root) {};
  \node (lamx) [above=1em of root] {$\lambda$};
  \node (lamy) [above left=1em of lamx] {$\lambda$};
  \node (appx) [above=2em of lamy] {$@$};
  \node (lamz) [right=2em of appx] {$\lambda$};
  \node (appy) [above=1em of lamz] {$@$};
  \node (vary) [above right=0.5em and 0em of appy] {};
  \node (dumy) [right=2em of lamz] {};
  \node (varx) [above right=4em and 1em of appx] {};
  \draw[->-] (lamx) to (root.center);
  \draw[->-] (lamy) to (lamx);
  \draw[->-,bend right=60] (appx) to (lamy);
  \draw[->-,bend left=45] (lamz) to (appx);
  \draw[->-,bend right=20,rounded corners=5pt] (lamy) to (dumy.center) to (vary.center) to (appy);
  \draw[->-,bend right=80,min distance=3cm,rounded corners=5pt] (lamx) to (varx.center) to [bend right=30,min distance=0em] (appx);
  \draw[->-,bend right=90] (appy) to (lamz);  
  \draw[->-,bend right=90] (lamz) to (appy);  
\end{tikzpicture}}}}
$$
\begin{theorem}\label{thm:trivalent-linear}
  To any rooted trivalent map $M$ with boundary of degree $k$ and $n$ trivalent vertices there is a unique linear lambda term with $k$ free variables, $p$ applications and $q$ abstractions whose underlying rooted trivalent map is $M$, for some $p+q = n$.
\end{theorem}
\begin{proof}
As sketched above, by induction on $n$.
Note that uniqueness relies on the fact that we define rooted trivalent maps with boundary as equipped with a fixed ordering on dangling edges, which determines the ordering of the free variables inside the context of the corresponding linear lambda term.
\end{proof}
\begin{corollary}\label{corr:bijection}
Rooted trivalent maps with boundary of degree $k$ and $n$ trivalent vertices are in one-to-one correspondence with linear lambda terms with $k$ free variables and $n$ total applications and abstractions.
Both families are counted by the generating function satisfying the functional-differential equation (\ref{eqn:lingf}).
\end{corollary}
\begin{example}\label{ex:bc}
The standard linear combinators \cite{curryhindleysendlin} $\Bcomb = \lambda x[\lambda y[\lambda z[x(yz)]]]$ and $\Ccomb = \lambda x[\lambda y[\lambda z[(xz)y]]]$ correspond to two different rooted embeddings of the $K_4$ graph (respectively, a planar embedding and a toric one):
$$
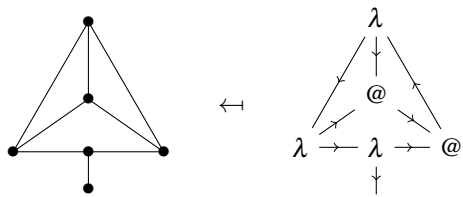
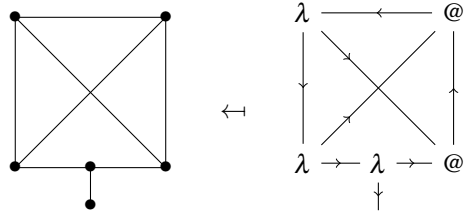
\vcenter{\hbox{\scalebox{1}{
\begin{tikzpicture}
\node (a) at (-1,0) {$\bullet$};
\node (b) at (0,1.73) {$\bullet$};
\node (c) at (1,0) {$\bullet$};
\node (d) at (0,0.7) {$\bullet$};
\node (e) at (0,0) {$\bullet$};
\node (root) at (0,-0.5) {$\bullet$};
\draw (a.center) to (b.center);
\draw (b.center) to (c.center);
\draw (c.center) to (e.center);
\draw (e.center) to (a.center);
\draw (e.center) to (root.center);
\draw (a.center) to (d.center);
\draw (b.center) to (d.center);
\draw (d.center) to (c.center);
\end{tikzpicture}
}}}
\quad\mapsfrom\quad
\vcenter{\hbox{\scalebox{1}{
\begin{tikzpicture}
\node (a) at (-1,0) {$\lambda$};
\node (b) at (0,1.73) {$\lambda$};
\node (c) at (1,0) {$@$};
\node (d) at (0,0.7) {$@$};
\node (e) at (0,0) {$\lambda$};
\node (root) at (0,-0.75) {};
\draw [->-] (b) to (a);
\draw [->-] (c) to (b);
\draw [->-] (e) to (c);
\draw [->-] (a) to (e);
\draw [->-] (e) to (root);
\draw [->-] (a) to (d);
\draw [->-] (b) to (d);
\draw [->-] (d) to (c);
\end{tikzpicture}
}}}
$$
$$
\vcenter{\hbox{\scalebox{1}{
\begin{tikzpicture}
\node (a) at (-1,0) {$\bullet$};
\node (b) at (-1,2) {$\bullet$};
\node (c) at (1,0) {$\bullet$};
\node (d) at (1,2) {$\bullet$};
\node (e) at (0,0) {$\bullet$};
\node (root) at (0,-0.5) {$\bullet$};
\draw (a.center) to (b.center);
\draw (b.center) to (c.center);
\draw (c.center) to (e.center);
\draw (e.center) to (a.center);
\draw (e.center) to (root.center);
\draw (a.center) to (d.center);
\draw (b.center) to (d.center);
\draw (d.center) to (c.center);
\end{tikzpicture}
}}}
\quad\mapsfrom\quad
\vcenter{\hbox{\scalebox{1}{
\begin{tikzpicture}
\node (a) at (-1,0) {$\lambda$};
\node (b) at (-1,2) {$\lambda$};
\node (c) at (1,0) {$@$};
\node (d) at (1,2) {$@$};
\node (e) at (0,0) {$\lambda$};
\node (root) at (0,-0.75) {};
\draw [->-] (b) to (a);
\draw [->-] (d) to (b);
\draw [->-] (e) to (c);
\draw [->-] (a) to (e);
\draw [->-] (e) to (root);
\draw [!->-=0.25] (a) to (d);
\draw [!->-=0.25] (b) to (c);
\draw [->-] (c) to (d);
\end{tikzpicture}
}}} 
$$
\end{example}
\begin{example}\label{ex:petersen}
A rooted embedding of the Petersen graph, and its corresponding linear lambda term ($\lambda a[\lambda b[\lambda c[\lambda d[\lambda e[a(\lambda f[c(e(b(d(f))))])]]]]]$):
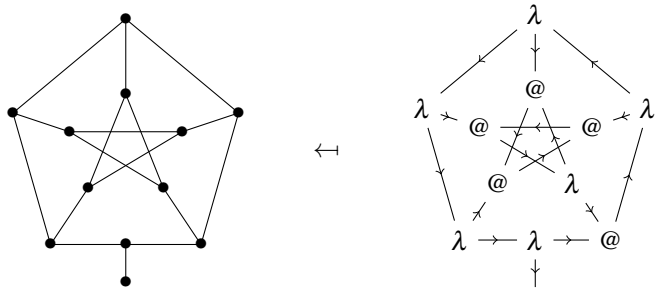
$$
\vcenter{\hbox{\scalebox{1}{
\begin{tikzpicture}
\node (0) at (-1, -0.5) {$\bullet$};
\node (root) at (0, -1) {$\bullet$};
\node (2) at (-1.5, 1.25) {$\bullet$};
\node (3) at (1, -0.5) {$\bullet$};
\node (4) at (0, 2.5) {$\bullet$};
\node (5) at (1.5, 1.25) {$\bullet$};
\node (6) at (0, -0.5) {$\bullet$};
\node (7) at (-0.75, 1) {$\bullet$};
\node (8) at (0, 1.5) {$\bullet$};
\node (9) at (0.75, 1) {$\bullet$};
\node (10) at (-0.5, 0.25) {$\bullet$};
\node (11) at (0.5, 0.25) {$\bullet$};
\draw (0.center) to (2.center);
\draw (5.center) to (3.center);
\draw (4.center) to (5.center);
\draw (2.center) to (4.center);
\draw (6.center) to (0.center);
\draw (3.center) to (6.center);
\draw (6.center) to (root.center);
\draw (4.center) to (8.center);
\draw (10.center) to (8.center);
\draw (9.center) to (10.center);
\draw (7.center) to (9.center);
\draw (11.center) to (7.center);
\draw (8.center) to (11.center);
\draw (5.center) to (9.center);
\draw (2.center) to (7.center);
\draw (10.center) to (0.center);
\draw (3.center) to (11.center);
\end{tikzpicture}}}}
\qquad\mapsfrom\qquad
\vcenter{\hbox{\scalebox{1}{
\begin{tikzpicture}
\node (0) at (-1, -0.5) {$\lambda$};
\node (root) at (0, -1.25) {};
\node (2) at (-1.5, 1.25) {$\lambda$};
\node (3) at (1, -0.5) {$@$};
\node (4) at (0, 2.5) {$\lambda$};
\node (5) at (1.5, 1.25) {$\lambda$};
\node (6) at (0, -0.5) {$\lambda$};
\node (7) at (-0.75, 1) {$@$};
\node (8) at (0, 1.5) {$@$};
\node (9) at (0.75, 1) {$@$};
\node (10) at (-0.5, 0.25) {$@$};
\node (11) at (0.5, 0.25) {$\lambda$};
\draw [->-] (2) to (0);
\draw [->-] (3) to (5);
\draw [->-] (5) to (4);
\draw [->-] (4) to (2);
\draw [->-] (0) to (6);
\draw [->-] (6) to (3);
\draw [->-] (6) to (root);
\draw [->-] (4) to (8);
\draw [->-] (8) to (10);
\draw [->-] (10) to (9);
\draw [->-] (9) to (7);
\draw [->-] (7) to (11);
\draw [->-] (11) to (8);
\draw [->-] (5) to (9);
\draw [->-] (2) to (7);
\draw [->-] (0) to (10);
\draw [->-] (11) to (3);
\end{tikzpicture}}}}
$$
\end{example}

\section{Indecomposable linear lambda terms and the 4CT}
\label{sec:4ct}

Recall that a \emph{bridge} in a connected graph is any edge whose removal disconnects the graph.
Among the first five non-trivial closed rooted trivalent maps, exactly three of them contain bridges (we do not count the outgoing root edge as a bridge):
$$
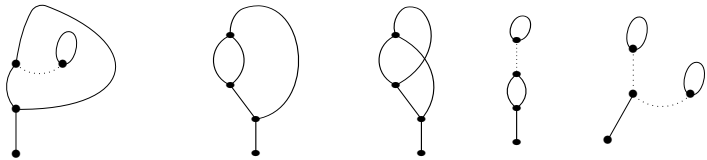
\vcenter{\hbox{\scalebox{0.8}{
\begin{tikzpicture}
  \node (root) {$\bullet$};
  \node (lamx) [above=1em of root] {$\bullet$};
  \node (appx) [above=1em of lamx] {$\bullet$};
  \node (varx) [above right=2em and 0em of appx] {};
  \node (lamy) [right=1em of appx] {$\bullet$};
  \node (vary) [above=0.6em of lamy] {};
  \draw[] (lamx.center) to (root.center);
  \draw[bend right=45] (appx.center) to (lamx.center);
  \draw[bend left=45,dotted] (lamy.center) to (appx.center);
  \draw[bend right=80,rounded corners=5pt,min distance=2cm] (lamx.center) to (varx.center) to [bend right=10,min distance=0em] (appx.center);
  \draw[bend right=90] (lamy.center) to (vary) to (lamy.center); 
\end{tikzpicture}}}}
\qquad
\vcenter{\hbox{\scalebox{0.8}[0.6]{
\begin{tikzpicture}
  \node (root) {$\bullet$};
  \node (lamx) [above=1em of root] {$\bullet$};
  \node (lamy) [above left=1em and 0em of lamx] {$\bullet$};
  \node (appx) [above=2em of lamy] {$\bullet$};
  \node (varx) [above right=1em and 0em of appx] {};
  \draw[] (lamx.center) to (root.center);
  \draw[] (lamy.center) to (lamx.center);
  \draw[bend right=60] (appx.center) to (lamy.center);
  \draw[bend right=45] (lamy.center) to (appx.center);
  \draw[bend right=80,rounded corners=5pt] (lamx.center) to (varx.center) to [bend right=30] (appx.center);
\end{tikzpicture}}}}
\qquad
\vcenter{\hbox{\scalebox{0.8}[0.6]{
\begin{tikzpicture}
  \node (root) {$\bullet$};
  \node (lamx) [above=1em of root] {$\bullet$};
  \node (lamy) [above left=1em and 0em of lamx] {$\bullet$};
  \node (appx) [above=2em of lamy] {$\bullet$};
  \node (varx) [above right=1em and 0em of appx] {};
  \draw[] (lamx.center) to (root.center);
  \draw[] (lamy.center) to (lamx.center);
  \draw[bend right=60] (appx.center) to (lamy.center);
  \draw[bend right=45,rounded corners=5pt] (lamy.center) to (varx.center) to (appx.center);
  \draw[bend right=45] (lamx.center) to (appx.center);
\end{tikzpicture}}}}
\qquad
\vcenter{\hbox{\scalebox{0.8}[0.6]{
\begin{tikzpicture}
  \node (root) {$\bullet$};
  \node (lamx) [above=1em of root] {$\bullet$};
  \node (appx) [above=1em of lamx] {$\bullet$};
  \node (lamy) [above=1em of appx] {$\bullet$};
  \node (vary) [above=0.6em of lamy] {};
  \draw[] (lamx.center) to (root.center);
  \draw[,bend right=45] (appx.center) to (lamx.center);
  \draw[bend right=45] (lamx.center) to (appx.center);
  \draw[dotted] (lamy.center) to (appx.center);
  \draw[bend right=90] (lamy.center) to (vary) to (lamy.center); 
\end{tikzpicture}}}}
\qquad
\vcenter{\hbox{\scalebox{0.8}{
\begin{tikzpicture}
  \node (root) {$\bullet$};
  \node (appx) [above right=1em and 0em of root] {$\bullet$};
  \node (lamy) [right=1.5em of appx] {$\bullet$};
  \node (vary) [above=0.6em of lamy] {};
  \node (lamx) [above=1em of appx] {$\bullet$};
  \node (varx) [above=0.6em of lamx] {};
  \draw[] (appx.center) to (root.center);
  \draw[bend left=45,dotted] (lamy.center) to (appx.center);
  \draw[dotted] (lamx.center) to (appx.center);
  \draw[bend right=90] (lamx.center) to (varx) to (lamx.center); 
  \draw[bend right=90] (lamy.center) to (vary) to (lamy.center); 
\end{tikzpicture}}}}
$$
If we look at the corresponding string diagrams,
$$
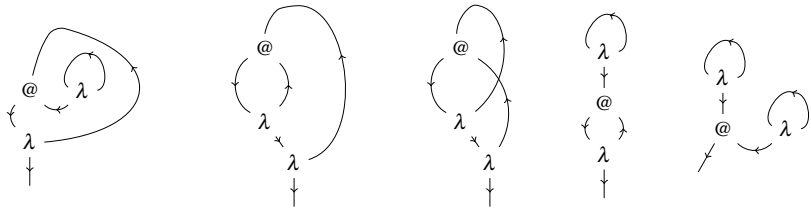
\vcenter{\hbox{\scalebox{0.8}{
\begin{tikzpicture}
  \node (root) {};
  \node (lamx) [above=1em of root] {$\lambda$};
  \node (appx) [above=1em of lamx] {$@$};
  \node (varx) [above right=2em and 0em of appx] {};
  \node (lamy) [right=1em of appx] {$\lambda$};
  \node (vary) [above=0.6em of lamy] {};
  \draw[->-] (lamx) to (root.center);
  \draw[->-,bend right=45] (appx) to (lamx);
  \draw[->-,bend left=45] (lamy) to (appx);
  \draw[->-,bend right=80,rounded corners=5pt,min distance=2cm] (lamx) to (varx.center) to [bend right=10,min distance=0em] (appx);
  \draw[->-,bend right=90] (lamy) to (vary) to (lamy); 
\end{tikzpicture}}}}
\qquad
\vcenter{\hbox{\scalebox{0.8}{
\begin{tikzpicture}
  \node (root) {};
  \node (lamx) [above=1em of root] {$\lambda$};
  \node (lamy) [above left=0.5em and 0em of lamx] {$\lambda$};
  \node (appx) [above=2em of lamy] {$@$};
  \node (varx) [above right=1em and 0em of appx] {};
  \draw[->-] (lamx) to (root.center);
  \draw[->-] (lamy) to (lamx);
  \draw[->-,bend right=60] (appx) to (lamy);
  \draw[->-,bend right=45] (lamy) to (appx);
  \draw[->-,bend right=80,rounded corners=5pt] (lamx) to (varx.center) to [bend right=30] (appx);
\end{tikzpicture}}}}
\qquad
\vcenter{\hbox{\scalebox{0.8}{
\begin{tikzpicture}
  \node (root) {};
  \node (lamx) [above=1em of root] {$\lambda$};
  \node (lamy) [above left=0.5em and 0em of lamx] {$\lambda$};
  \node (appx) [above=2em of lamy] {$@$};
  \node (varx) [above right=1em and 0em of appx] {};
  \draw[->-] (lamx) to (root.center);
  \draw[->-] (lamy) to (lamx);
  \draw[->-,bend right=60] (appx) to (lamy);
  \draw[->-,bend right=45,rounded corners=5pt] (lamy) to (varx.center) to (appx);
  \draw[->-,bend right=45] (lamx) to (appx);
\end{tikzpicture}}}}
\qquad
\vcenter{\hbox{\scalebox{0.8}{
\begin{tikzpicture}
  \node (root) {};
  \node (lamx) [above=1em of root] {$\lambda$};
  \node (appx) [above=1em of lamx] {$@$};
  \node (lamy) [above=1em of appx] {$\lambda$};
  \node (vary) [above=0.6em of lamy] {};
  \draw[->-] (lamx) to (root.center);
  \draw[->-,bend right=45] (appx) to (lamx);
  \draw[->-,bend right=45] (lamx) to (appx);
  \draw[->-] (lamy) to (appx);
  \draw[->-,bend right=90] (lamy) to (vary) to (lamy); 
\end{tikzpicture}}}}
\qquad
\vcenter{\hbox{\scalebox{0.8}{
\begin{tikzpicture}
  \node (root) {};
  \node (appx) [above right=1em and 0em of root] {$@$};
  \node (lamy) [right=1.5em of appx] {$\lambda$};
  \node (vary) [above=0.6em of lamy] {};
  \node (lamx) [above=1em of appx] {$\lambda$};
  \node (varx) [above=0.6em of lamx] {};
  \draw[->-] (appx) to (root.center);
  \draw[->-,bend left=45] (lamy) to (appx);
  \draw[->-] (lamx) to (appx);
  \draw[->-,bend right=90] (lamx) to (varx) to (lamx); 
  \draw[->-,bend right=90] (lamy) to (vary) to (lamy); 
\end{tikzpicture}}}}
$$
we see that each of the bridges corresponds to a wire oriented towards an $@$-node (either in function or in argument position), and that it sends a \emph{closed subterm} of the underlying linear lambda term (in these three cases an identity term $\Icomb$) to that $@$-node.
\begin{definition}\label{def:subterm}
Let $(\Gamma,t)$ be a linear lambda term.
A \definand{subterm} of $(\Gamma,t)$ is a linear lambda term $(\Delta,u)$ that appears in the derivation of $\Gamma \vdash t$.
Explicitly:
\begin{itemize}
\item $(\Gamma,t)$ is a subterm of itself;
\item if $t = t_1(t_2)$ for some $\Gamma_1 \vdash t_1$ and $\Gamma_2 \vdash t_2$, then every subterm $(\Delta,u)$ of $(\Gamma_1,t_1)$ or $(\Gamma_2,t_2)$ is also a subterm of $(\Gamma,t)$; and
\item if $t = \lambda x.t_1$ for some $\Gamma,x \vdash t_1$, then every subterm $(\Delta,u)$ of $((\Gamma,x),t_1)$ is also a subterm of $(\Gamma,t)$.
\end{itemize}
We refer to all the subterms of $(\Gamma,t)$ other than $(\Gamma,t)$ itself as \definand{proper subterms}.
\end{definition}
\begin{definition}\label{def:indecomposable}
A linear lambda term is said to be \definand{decomposable} if it has a closed proper subterm, and \definand{indecomposable} otherwise.
\end{definition}
\begin{proposition}\label{claim:bridgeless}
A closed rooted trivalent map is bridgeless if and only if the corresponding closed linear lambda term is indecomposable.
\end{proposition}
\noindent
To prove this claim, let's first recall the notion of the \emph{lambda lifting} of a term with free variables.
\begin{definition}
Let $(\Gamma,t)$ be a linear lambda term, where $\Gamma = (x_1,\dots,x_k)$.
The \definand{lambda lifting} of $(\Gamma,t)$ is the closed linear lambda term $\llift{\Gamma}t \defeq \lambda x_1[\cdots\lambda x_k[t]\cdots]$.
\end{definition}
\begin{proof}[Proof of \Cref{claim:bridgeless}]
By \Cref{thm:trivalent-linear}, it suffices to do an induction over linear lambda terms, and check whether the underlying rooted trivalent map of their lambda lifting contains a bridge (and again, we do not count the outgoing root edge itself as a bridge).
From examination of the root-deletion procedure, it is immediate that the only way of potentially introducing a bridge is by using an $@$-node to form an application $t = t_1(t_2)$, so we just have to check whether or not removing the edge corresponding to either the function port ($t_1$) or argument port ($t_2$) of the $@$-node disconnects the diagram of $\llift{\Gamma}t$.
Well, if $t_i$ has a free variable $x$, then the subdiagram rooted at $t_i$ will remain connected to the root port of $\llift{\Gamma}t$, by a path running through the root port of the $\lambda$-node for $x$.
Hence, the edge corresponding to $t_i$ is a bridge in the underlying rooted trivalent map of $\llift{\Gamma}t$ just in case $t_i$ is closed.
\end{proof}
\noindent
This analysis immediately suggests a way of enumerating bridgeless rooted trivalent maps.
\begin{proposition}
The generating function $\GFlinind(z,x)$ counting indecomposable linear lambda terms by size (= number of applications and abstractions) and number of free variables satisfies the following functional-differential equation:
\begin{equation}
\GFlinind(z,x) = x + z(\GFlinind(z,x)-\GFlinind(z,0))^2 + z\frac{\partial}{\partial x}\GFlinind(z,x)
\label{eqn:lingf-ind}
\end{equation}
In particular, the OGF $$\GFlinind(z,0) = z + 2z^3 + 20z^5 + 352 z^7 + 8624^9 + 266784 z^{11} + \dots$$
counts closed indecomposable linear lambda terms by size, as well as closed bridgeless rooted trivalent maps (on oriented surfaces of arbitrary genus) by number of trivalent vertices.
\end{proposition}
\noindent
Now, let us say that a linear lambda term is \definand{planar} just in case its underlying rooted trivalent map is planar.
Planar lambda terms have the special property that after we've fixed the convention for ordering wires around $@$-nodes and $\lambda$-nodes, there is always exactly one planar term with any given underlying tree of applications and abstractions (see \cite[Zeilberger and Giorgetti (2015, \S2)]{zg2015rpmnpt}).
This makes them easy to count, and the following ``discrete analogues'' of (\ref{eqn:lingf}) and (\ref{eqn:lingf-ind}) define the two-variable generating functions for planar lambda terms and indecomposable planar lambda terms, respectively:
\begin{align}
\GFpla(z,x) &= x + z\GFpla(z,x)^2 + z\frac{\GFpla(z,x) - \GFpla(z,0)}{x} \label{eqn:plagf} \\
\GFplaind(z,x) &= x + z(\GFplaind(z,x)-\GFplaind(z,0))^2 + z\frac{\GFplaind(z,x) - \GFplaind(z,0)}{x} \label{eqn:plagf-ind}
\end{align}
In particular, the OGF $\GFplaind(z,0) = z + z^3 + 4z^5 + 24 z^7 + 176 z^9 + 1456 z^{11} + \dots$ counts rooted bridgeless planar trivalent maps by number of trivalent vertices, as originally enumerated by Tutte \shortcite{tutte1962} (OEIS \href{https://oeis.org/A000309}{A000309}; keep in mind that we define closed rooted trivalent maps to contain one extra trivalent vertex relative to the classical definition, cf.~\Cref{prop:bijroot}).

Bridgeless planar trivalent maps are closely related to the Four Color Theorem: by Tait's well-known reduction \cite{thomas98}, the statement that every bridgeless planar map has a proper 4-coloring of its faces is equivalent to the statement that every bridgeless planar trivalent map has a proper 3-coloring of its edges, i.e., a labelling of the edges by colors in $\set{R,G,B}$ such that every vertex has the form
$$
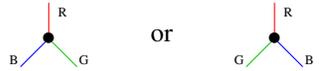
\vcenter{\hbox{\begin{tikzpicture}
\node (R) at (0,0.6) {};
\node (G) at (0.5,-0.5) {};
\node (B) at (-0.5,-0.5) {};
\node (v) at (0,0) {$\bullet$};
\path
  (R) edge [red] node [near start,right,black] {\tiny R} (v.center)
  (G) edge [darkgreen] node [near start,right,black] {\tiny G} (v.center)
  (B) edge [blue] node [near start,left,black] {\tiny B} (v.center);
\node (vp) at (0,0) {$\bullet$};
\end{tikzpicture}}}
\qquad\text{or}\qquad
\vcenter{\hbox{\begin{tikzpicture}
\node (R) at (0,0.6) {};
\node (B) at (0.5,-0.5) {};
\node (G) at (-0.5,-0.5) {};
\node (v) at (0,0) {$\bullet$};
\path
  (R) edge [red] node [near start,right,black] {\tiny R} (v.center)
  (B) edge [blue] node [near start,right,black] {\tiny B} (v.center)
  (G) edge [darkgreen] node [near start,left,black] {\tiny G} (v.center);
\node (vp) at (0,0) {$\bullet$};
\end{tikzpicture}}}.
$$
For the purposes of coloring, there is little difference between rooted and unrooted maps: without loss of generality, it suffices to root a trivalent map $M$ arbitrarily at some edge by splitting it with a trivalent vertex, assign both halves of that edge the same arbitrary color
($
\vcenter{\hbox{\scalebox{0.6}
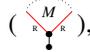
{
\begin{tikzpicture}
\node (root) at (0,-0.3) {$\bullet$};
\node (Rl) at (0.5,0.5) {};
\node (Rr) at (-0.5,0.5) {};
\node (v) at (0,0) {$\bullet$};
\node (map) at (0,0.5) {$M$};
\path
  (root.center) edge (v.center)
  (Rl.center) edge [red] node [near end,right,black] {\tiny R} (v.center)
  (Rr.center) edge [red] node [near end,left,black] {\tiny R} (v.center)
  (Rl.center) edge [dotted,bend right=90] (Rr.center);
\node (vp) at (0,0) {$\bullet$};
\end{tikzpicture}}}}
$), 
and then look for a proper 3-coloring of the remaining edges.

It seems that the problem of 3-coloring the edges of a rooted trivalent map may be naturally formulated as a \emph{typing} problem in linear lambda calculus.
Typing for linear lambda calculus is standardly defined by the following rules:
\begin{equation}\label{eqn:typing}
\infer{x:X \vdash x:X}{}
\quad
\infer{\Gamma,\Delta \vdash t(u) : Y}{\Gamma \vdash t : \impL[Y]{X} & \Delta \vdash u : X}
\quad
\infer{\Gamma \vdash \lambda x[t] : \impL[Y]{X}}{\Gamma,x:X \vdash t : Y}
\quad
\quad
\infer{\Gamma, x:X,y:Y,\Delta \vdash t:Z}{\Gamma, y:Y,x:X,\Delta \vdash t : Z}
\end{equation}
General types ($X,Y,\dots$) are built up from some set of type variables ($\alpha,\beta,\dots$) using only implication ($\impL[Y]{X}$), and the typing judgment $x_1:X_1,\dots,x_k:X_k \vdash t:Y$ expresses that the given term $t$ has type $Y$ assuming that the free variables have the prescribed types $X_1,\dots,X_k$.
In this way, closed linear lambda terms can be seen as proofs of tautologies in a very weak, purely implicative logic (sometimes called BCI logic, after the combinators $\Bcomb$, $\Ccomb$, and $\Icomb$).
Moreover, planar lambda terms are typable without using the rightmost rule, resulting in an even weaker logic.

The standard typing rules can also be expressed concisely using string diagrams, where they correspond to the following conditions for annotating the wires by types (cf.~\cite{mairson2002dilbert,z2015balanced}):
$$
\vcenter{\hbox{\begin{tikzpicture}
  \node (fn) {};
  \node (app) [below=1em of fn] {$@$};
  \node (cont) [below left=1em of app] {};
  \node (arg) [below right=1em of app] {};
  \path
    (fn.center) edge [->-] node [near start,right] {\tiny$\impL[Y]{X}$} (app)
    (app) edge [->-] node [near end,left] {\tiny$Y$} (cont.center)
    (arg.center) edge [->-] node [near start,right] {\tiny$X$} (app);  
\end{tikzpicture}}}
\qquad\qquad
\vcenter{\hbox{\begin{tikzpicture}
  \node (root) {};
  \node (lam) [above=1em of root] {$\lambda$};
  \node (var) [above right=1em of lam] {};
  \node (body) [above left=1em of lam] {};
  \path
     (lam) edge [->-] node [near end,right] {\tiny$\impL[Y]{X}$} (root.center)
     (lam) edge [->-] node [near end,right] {\tiny$X$} (var.center)
     (body.center) edge [->-] node [near start,left=1pt] {\tiny$Y$} (lam);  
\end{tikzpicture}}}
$$
In this form, the connection to edge-coloring is more suggestive.
Indeed, we can use a specific interpretation of types in order to obtain a new reformulation of the map coloring theorem (cf.~\cite{penrose,kauffman1990,barnatan97}):

Recall that the \emph{Klein Four Group} can be defined as a group whose underlying set has four elements $\mathbb{V} = \set{1,R,G,B}$, with unit element 1 and the following multiplication table for non-unit elements:
$$
\begin{array}{c|ccc}
& R & G & B \\
\hline
R & 1 & B & G \\
G & B & 1 & R\\
B & G & R & 1
\end{array}
$$
Observe that the product operation of the Klein Four Group is commutative $xy = yx$, and that every element is its own inverse $x^{-1} = x$.
\begin{definition}
We write $\vdash_{\mathbb V}$ for the typing judgment induced from (\ref{eqn:typing}) by restricting types to elements of the Klein Four Group $\mathbb{V}$ and interpreting implication by $x \ImpL y = yx^{-1} = xy$ for all $x,y \in \mathbb{V}$. 
A \definand{3-typing} of a linear lambda term $t$ with free variables $x_1,\dots,x_k$ is defined as a derivation of the typing judgment $x_1:X_1,\dots,x_k:X_k \vdash_{\mathbb{V}} t:Y$ for some $X_1,\dots,X_k$ and $Y$ in $\mathbb{V}$.
The 3-typing is said to be \definand{proper} if no proper subterm of $t$ is assigned type 1.
\end{definition}
\begin{theorem}[Reformulation of 4CT]
{\it Every planar indecomposable linear lambda term has a proper 3-typing.}
\end{theorem}
\begin{example}\label{ex:3color-b}
The $\Bcomb$ combinator (\Cref{ex:bc}) has most general type
$$\Bcomb : \impL[({\impL[{(\impL[\gamma]{\alpha})}]{(\impL[\beta]{\alpha})}})]{(\impL[\gamma]{\beta})}$$
corresponding to the following formal typing derivation with type variables $\alpha,\beta,\gamma$:
$$
\infer{\vdash \lambda x[\lambda y[\lambda z[x(yz)]]] : \impL[({\impL[{(\impL[\gamma]{\alpha})}]{(\impL[\beta]{\alpha})}})]{(\impL[\gamma]{\beta})}}{
\infer{x :\impL[\gamma]{\beta} \vdash \lambda y[\lambda z[x(yz)]] : \impL[{(\impL[\gamma]{\alpha})}]{(\impL[\beta]{\alpha})}}{
\infer{x :\impL[\beta]{\alpha}, y : \impL[\beta]{\alpha} \vdash \lambda z[x(yz)] : \impL[\gamma]{\alpha}}{
\infer{x :\impL[\beta]{\alpha}, y : \impL[\beta]{\alpha}, z : \alpha \vdash x(yz) : \gamma}{
 \infer{x :\impL[\gamma]{\beta} \vdash x :\impL[\gamma]{\beta}}{} &
 \infer{y : \impL[\beta]{\alpha}, z : \alpha \vdash y(z) : \beta}{
   \infer{y : \impL[\beta]{\alpha} \vdash y : \impL[\beta]{\alpha}}{} &
   \infer{z : \alpha \vdash z : \alpha}{}
}}}}}
$$
Instantiating $\alpha = R$, $\beta = B$, and $\gamma = G$, we obtain a proper 3-typing of the combinator:
$$
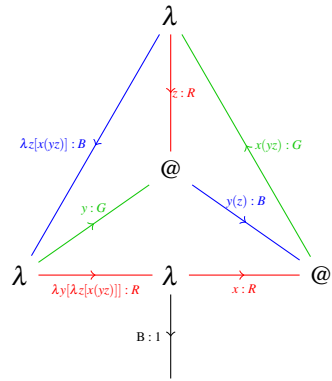
\vcenter{\hbox{\scalebox{1}{
\begin{tikzpicture}[scale=2]
\node (a) at (-1,0) {$\lambda$};
\node (b) at (0,1.73) {$\lambda$};
\node (c) at (1,0) {$@$};
\node (d) at (0,0.7) {$@$};
\node (e) at (0,0) {$\lambda$};
\node (root) at (0,-0.75) {};
\draw [->-,blue] (b) to node [left] {\tiny$\lambda z[x(yz)]:B$} (a);
\draw [->-,color=darkgreen] (c) to node [right] {\tiny$x(yz):G$} (b);
\draw [->-,red] (e) to node [below] {\tiny$x:R$} (c);
\draw [->-,red] (a) to node [below] {\tiny$\lambda y[\lambda z[x(yz)]]:R$} (e);
\draw [->-] (e) to node [left] {\tiny$\Bcomb : 1$} (root);
\draw [->-,color=darkgreen] (a) to node [above=0.0]{\tiny$y:G$} (d);
\draw [->-,red] (b) to node [right=-0.1]{\tiny$z:R$} (d);
\draw [->-,blue] (d) to node [above=0.07]{\tiny$y(z):B$} (c);
\end{tikzpicture}
}}}
$$
\end{example}

\noindent
    {\bf Acknowledgments.}
I am grateful to Alain Giorgetti for introducing me to the idea of Tutte decomposition and for emphasizing its importance as a way of representing rooted maps.
Special thanks to Andrej Bauer for inviting me to give a talk at the University of Ljubljana Foundations Seminar, which was the original impetus for writing this article.
Thanks also to Doron Zeilberger for helpful comments on an earlier draft, as well as to the anonymous referees.
Finally, this work has been funded by the ERC Advanced Grant ProofCert.

\bibliographystyle{abbrvnat}

\begin{thebibliography}{}

\bibitem[\protect\citename{Barendregt, }1984]{barendregt1984}
Barendregt,~H.~P. (1984) {\em The Lambda Calculus: Its Syntax and Semantics}, Studies in Logic 103, second, revised edition, North-Holland, Amsterdam.

\bibitem[\protect\citename{Bar-Natan, }1997]{barnatan97}
Bar-Natan,~D. (1997) Lie algebras and the four color theorem, {\em Combinatorica} 17, 43--52.

\bibitem[\protect\citename{Bodini \emph{et al.}, }2013]{bodini-et-al}
Bodini,~O., Gardy,~D., and Jacquot,~A. (2013) Asymptotics and random sampling for BCI and BCK lambda terms. {\em Theoretical Computer Science}, 502:227--238.

\bibitem[\protect\citename{Curry \emph{et al.}, }1972]{curryhindleysendlin}
Curry,~H.~B., Hindley,~J.~R., and Seldin,~J.~P. (1972) {\em Combinatory Logic}, vol. II, North-Holland.

\bibitem[\protect\citename{Hyland, }2013]{hyland-lambda-calculus}
Hyland,~M. (2013) Classical lambda calculus in modern dress. {\em Mathematical Structures in Computer Science}.  In special issue dedicated to Corrado Böhm for his 90th birthday, 1--20.

\bibitem[\protect\citename{Jacobs, }1993]{jacobs1993}
Jacobs,~B. (1993) Semantics of lambda-I and of other substructural lambda calculi. In \emph{Typed Lambda Calculus and Applications}, Springer LNCS 664, 195--208.

\bibitem[\protect\citename{Jones and Singerman, }1978]{jones-singerman}
Jones,~G.~A. and Singerman,~D. (1978) Theory of maps on orientable surfaces. {\em Proceedings of the London Mathematical Society}, 37:273--307.

\bibitem[\protect\citename{Jones and Singerman, }1994]{jones-singerman94schneps}
Jones,~G.~A. and Singerman,~D. (1994) Maps, hypermaps, and triangle groups.  In {\em The Grothendieck Theory of Dessins d'Enfants}, L. Schneps (ed.), London Mathematical Society Lecture Note Series 200, Cambridge University Press.

\bibitem[\protect\citename{Joyal and Street, }1993]{joyal-street-i}
Joyal,~A. and Street,~R. The geometry of tensor calculus I. {\em Advances in Mathematics}, 102:20--78.

\bibitem[\protect\citename{Kauffman, }1990]{kauffman1990}
Kauffman,~L.H. (1990) Map Coloring and the Vector Cross Product. {\em Journal of Combinatorial Theory B} 48:145--154.

\bibitem[\protect\citename{Knuth, }1970]{knuth1970}
Knuth,~D.E. (1970) Examples of formal semantics. In \emph{Symposium on Semantics of Algorithmic Languages}, E. Engeler (ed.), Lecture Notes in Mathematics 188, Springer.

\bibitem[\protect\citename{Lando and Zvonkin, }2004]{landozvonkin}
Lando,~S.~K. and Zvonkin,~A.K. (2004) {\em Graphs on Surfaces and Their Applications}, Encyclopaedia of Mathematical Sciences 141, Springer-Verlag.

\bibitem[\protect\citename{Mairson, }2002]{mairson2002dilbert}
Mairson,~H.G. (2002) From Hilbert Spaces to Dilbert Spaces: Context Semantics Made Simple. In {\em Proceedings of the 22nd Conference on Foundations of Software Technology and Theoretical Computer Science}, 2--17, Kanpur, India.

\bibitem[\protect\citename{Mairson, }2004]{mairson2004}
Mairson,~H.G. (2004) Linear lambda calculus and PTIME-completeness. \emph{Journal of Functional Programming}, 14:6.

\bibitem[\protect\citename{OEIS, }2016]{oeis}
OEIS Foundation Inc. (2016) The On-Line Encyclopedia of Integer Sequences.

\bibitem[\protect\citename{Penrose, }1971]{penrose}
Penrose,~R. (1971) Applications of Negative Dimensional Tensors.  In D. Welsh (ed.), {\em Combinatorial Mathematics and its Application}, 221--243.

\bibitem[\protect\citename{Scott, }1980]{scott1980}
Scott,~D.~S. (1980) Relating theories of the $\lambda$-calculus.   In {\em To H.B. Curry: Essays on Combinatory Logic, Lambda-Calculus and Formalism} (eds. Hindley and Seldin), Academic Press, 403--450.

\bibitem[\protect\citename{Seely, }1987]{seely1987}
Seely,~R.~A.~G. (1987) Modelling Computations: A 2-Categorical Framework. In {\em Proceedings of the Second Annual IEEE Symposium on Logic in Computer Science}, 65--71, Ithaca, NY, USA.

\bibitem[\protect\citename{Selinger, }2011]{selinger-survey}
Selinger,~P. (2011)  A survey of graphical languages for monoidal categories. In {\em New Structures for Physics} (ed. Bob Coecke), Springer Lecture Notes in Physics 813, 289--355.

\bibitem[\protect\citename{Stay, }2013]{stay2013}
Stay,~M. (2013) Compact closed bicategories. arXiv:1301.1053.

\bibitem[\protect\citename{Thomas, }1998]{thomas98}
Thomas,~R. (1998) An Update on the Four-Color Theorem. {\em Notices of the American Mathematical Society} 45:7, 848--859.

\bibitem[\protect\citename{Tutte, }1962]{tutte1962}
Tutte,~W.~T. (1962) A census of Hamiltonian polygons. {\em Canadian Journal of Mathematics}, 14:402--417.

\bibitem[\protect\citename{Tutte, }1968]{tutte1968}
Tutte,~W.~T. (1968)  On the enumeration of planar maps. {\em Bulletin of the American Mathematical Society}, 74:64--74.

\bibitem[\protect\citename{Vidal, }2010]{vidal2010}
Vidal,~S. (2010) {\em Groupe Modulaire et Cartes Combinatoires: G\'en\'eration et Comptage.} PhD thesis, Universit\'e Lille I, France (July).


\bibitem[\protect\citename{Zeilberger and Giorgetti, }2015]{zg2015rpmnpt}
Zeilberger,~N. and Giorgetti,~A. (2015) A correspondence between rooted planar maps and normal planar lambda terms. {\em Logical Methods in Computer Science}, 11(3:22):1--39.

\bibitem[\protect\citename{Zeilberger, }2015a]{z2015balanced}
Zeilberger,~N. (2015) Balanced polymorphism and linear lambda calculus. Talk at TYPES 2015, Tallinn, Estonia (18 May).

\bibitem[\protect\citename{Zeilberger, }2015b]{z2015counting}
Zeilberger,~N. (2015) Counting isomorphism classes of $\beta$-normal linear lambda terms.
  arXiv:1509.07596
  (25 September)

\end{thebibliography}

\end{document}